\documentclass[conference]{IEEEtran}
\IEEEoverridecommandlockouts

\usepackage{amsmath,amsfonts}
\usepackage{algorithmic}

\usepackage{array}

\usepackage[caption=false,font=footnotesize]{subfig}

\usepackage{textcomp}
\usepackage{stfloats}
\usepackage{url}
\usepackage{verbatim}
\usepackage{graphicx}
\usepackage{balance}

\usepackage{cite}

\usepackage[ruled,linesnumbered,vlined]{algorithm2e}
\usepackage{amsthm}

\usepackage{makecell}
\usepackage{multirow}
\usepackage{tikz}
\usetikzlibrary{plotmarks}
\usetikzlibrary{patterns}
\usepackage{pgfplots}
\usepackage{diagbox}
\usepackage[a-2b]{pdfx}
\usepackage{pifont}
\usepackage{ragged2e} 
\usepackage{amssymb}
\hypersetup{hidelinks}

\newcommand{\Acom}{{ByteDance}}

\newtheorem{definition}{Definition}
\newtheorem{theorem}{Theorem}

\newtheorem{problem}{Problem}
\newtheorem{assumption}{Assumption}
\newtheorem{example}{Example}
\newtheorem{observation}{Observation}
\newtheorem{lemma}{Lemma}

\newcolumntype{M}[1]{>{\centering\arraybackslash}m{#1}}
\newcolumntype{L}[1]{>{\RaggedRight\arraybackslash}m{#1}}

\definecolor{c1}{RGB}{49,87,149}
\definecolor{c2}{RGB}{221,49,24}
\definecolor{c3}{RGB}{107,186,85}
\definecolor{c4}{RGB}{60,151,178}
\definecolor{c5}{RGB}{243,178,147}
\definecolor{c6}{RGB}{142,96,160}
\definecolor{c7}{RGB}{91, 8, 136}
\definecolor{c8}{RGB}{232,34,45}
\definecolor{c9}{RGB}{0,24,113}
\definecolor{c10}{RGB}{255,88,93}

\makeatletter
\AtBeginDocument{%
  \renewcommand{\baselinestretch}{0.975}\normalsize
  \let\ACM@origbaselinestretch\baselinestretch
}

\pgfplotsset{compat=1.18,
        /pgfplots/ybar legend/.style={
        /pgfplots/legend image code/.code={
        \draw [#1,yshift=-0.4em] rectangle (0.7cm,0.3cm);
        },
    }
}

\def\BibTeX{{\rm B\kern-.05em{\sc i\kern-.025em b}\kern-.08em
    T\kern-.1667em\lower.7ex\hbox{E}\kern-.125emX}}
\begin{document}

\title{Maintaining Leiden Communities in Large Dynamic Graphs
(Full Version)
}

\makeatletter
\newcommand{\linebreakand}{%
  \end{@IEEEauthorhalign}
  \hfill\mbox{}\par
  \mbox{}\hfill\begin{@IEEEauthorhalign}
}
\makeatother

\author{
\IEEEauthorblockN{Chunxu Lin}
\IEEEauthorblockA{
\textit{CUHK-Shenzhen}\\
Shenzhen, China \\
chunxulin1@link.cuhk.edu.cn}
\and
\IEEEauthorblockN{Yumao Xie}
\IEEEauthorblockA{
\textit{CUHK-Shenzhen}\\
Shenzhen, China \\
yumaoxie@link.cuhk.edu.cn}
\and
\IEEEauthorblockN{Yixiang Fang}
\IEEEauthorblockA{
\textit{CUHK-Shenzhen}\\
Shenzhen, China \\
fangyixiang@cuhk.edu.cn}

\linebreakand

\IEEEauthorblockN{Yongmin Hu}
\IEEEauthorblockA{\textit{ByteDance Inc.} \\
Hangzhou, China \\
huyongmin@bytedance.com}
\and
\IEEEauthorblockN{Yingqian Hu}
\IEEEauthorblockA{\textit{ByteDance Inc.} \\
Hangzhou, China \\
huyingqian@bytedance.com}
\and
\IEEEauthorblockN{Cheng Chen}
\IEEEauthorblockA{\textit{ByteDance Inc.} \\
Singapore \\
chencheng.sg@bytedance.com}
}

\maketitle

\begin{abstract}
    As a well-known community detection algorithm, Leiden has been widely used in various scenarios such as large language model (LLM) generation, anomaly detection, and biological analysis. 
    In these scenarios, the graphs are often large and dynamic, where vertices and edges are inserted and deleted frequently, so it is costly to obtain the updated communities by Leiden from scratch when the graph has changed.
    Recently, one work has attempted to study how to maintain Leiden communities in the dynamic graph, but it lacks a detailed theoretical analysis, and its algorithms are inefficient for large graphs.
    To address these issues, in this paper, we first theoretically show that the existing algorithms are relatively unbounded via the boundedness analysis (a powerful tool for analyzing incremental algorithms on dynamic graphs), and also analyze the memberships of vertices in communities when the graph changes.
    Based on theoretical analysis, we develop a novel efficient maintenance algorithm, called \textit{Hierarchical Incremental Tree Leiden} ({\texttt{HIT-Leiden}}), which effectively reduces the range of affected vertices by maintaining the connected components and hierarchical community structures.
    Comprehensive experiments on various datasets demonstrate the superior performance of \texttt{HIT-Leiden}. In particular, it achieves speedups of up to five orders of magnitude over existing methods.
\end{abstract}


\section{Introduction}
\label{sec:intro}


As one of the fundamental measures in network science, modularity \cite{newman2004finding} effectively measures the strength of division of a network into modules (also called communities).
Essentially, it captures the difference between the actual number of edges within a community and the expected number of such edges if connections were random.
By maximizing the modularity of a graph, it can reveal all the communities in the graph.
In Fig. \ref{fig:mod}\subref{fig:mod:org_graph}, for example, by maximizing the modularity of the graph, we can obtain two communities $C_1$ and $C_2$. 
As shown in the literature \cite{chakraborty2017metrics,su2022comprehensive}, graph communities have found a wide range of applications in recommendation systems, social marketing, and biological analysis.

\begin{figure}[t]
  \centering
  \subfloat[A static graph $G$]{
    \begin{minipage}[t]{0.215625\textwidth}
      \centering
      \includegraphics[width=0.818181818\linewidth]{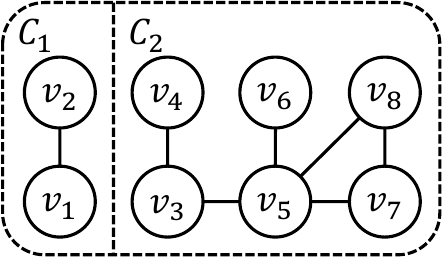}
    \end{minipage}
    \label{fig:mod:org_graph}
  }
  \hfil
  \subfloat[A dynamic graph $G^{'}$]{
    \begin{minipage}[t]{0.215625\textwidth}
      \centering
      \includegraphics[width=0.818181818\linewidth]{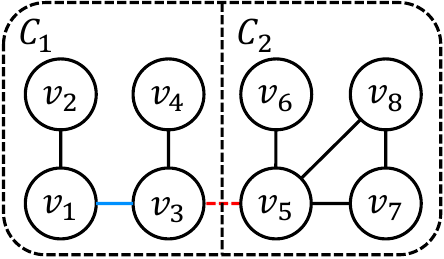}
    \end{minipage}
    \label{fig:mod:cur_partition}
  }
  \setlength{\abovecaptionskip}{-1pt}
  \setlength{\belowcaptionskip}{-2pt}
  \caption{Illustrating community maintenance, where $(v_1, v_3)$ is a newly inserted edge and $(v_3, v_5)$ is a newly deleted edge.}
  \label{fig:mod}
\end{figure}

One of the most popular community detection (CD) algorithms that uses modularity maximization is Louvain \cite{blondel2008fast}, which partitions a graph into disjoint communities.
As shown in Fig. \ref{fig:louvain_leiden}\subref{fig:louvain_leiden:louvain}, Louvain employs an iterative process with each iteration having two phases, called \textbf{\textit{movement}} and \textbf{\textit{aggregation}}, to adjust the community structure and improve modularity.
Specifically, in the movement phase, each vertex is relocated to a suitable community to maximize the modularity of the graph.
In the aggregation phase, all the vertices belonging to the same community are merged into a supervertex to form a supergraph for the next iteration.
Since a supervertex corresponds to a set of vertices, the communities of a graph naturally form a tree-like hierarchical structure. 
In practice, to balance modularity gains against the running time, users often limit Louvain to $P$ iterations, where $P$ is a pre-defined parameter.

\begin{figure}[t]
  \centering
  \subfloat[The process of the Louvain algorithm \cite{blondel2008fast}.]{
    \begin{minipage}[t]{0.45\textwidth}
      \centering
      \includegraphics[width=\linewidth]{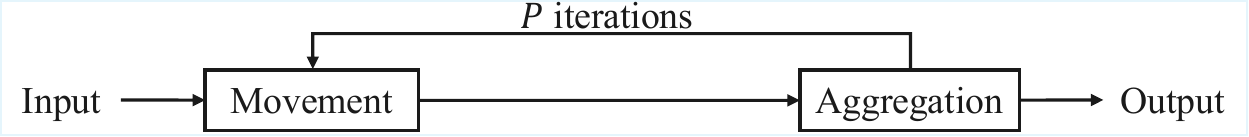}
    \end{minipage}
    \label{fig:louvain_leiden:louvain}
  }
  \hfil
  \subfloat[The process of the Leiden algorithm \cite{traag2019louvain}.]{
    \begin{minipage}[t]{0.45\textwidth}
      \centering
      \includegraphics[width=\linewidth]{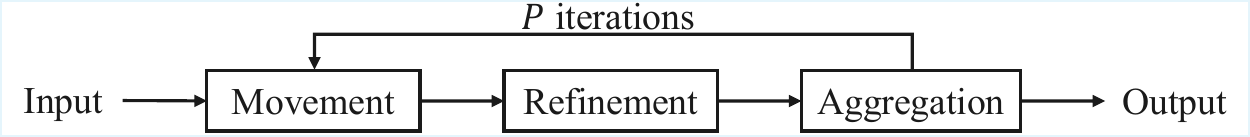}
    \end{minipage}
    \label{fig:louvain_leiden:leiden}
  }
  \setlength{\abovecaptionskip}{0pt}
  \setlength{\belowcaptionskip}{-2pt}
  \caption{Illustrating the Louvain and Leiden algorithms.}
  \label{fig:louvain_leiden}
\end{figure}

Despite its popularity, Louvain may produce communities that are internally disconnected. 
This typically occurs during the movement phase, where a vertex that serves as a bridge within a community may be moved to a different community that has stronger connections, thereby breaking the connectivity of the original community. 
To overcome this issue, Traag et al. \cite{traag2019louvain} proposed the \textit{Leiden} algorithm\footnote{As of June 2026, Leiden has received over 7,000 citations according to Google Scholar.},
which introduces an additional phase, called \textbf{\textit{refinement}}, between the movement and aggregation phases, as shown in Fig. \ref{fig:louvain_leiden}\subref{fig:louvain_leiden:leiden}.
Specifically, during the refinement phase, vertices explore merging with their neighbors within the same community to form sub-communities.
By adding this additional phase, Leiden produces communities with higher quality than Louvain, since its communities well preserve the connectivity.

As shown in the literature, Leiden has recently received plenty of attention because of its applications in many areas, including large language model (LLM) generation \cite{graphrag2025}, anomaly detection \cite{wang2023removing}, and biological analysis \cite{haney2024apoe4}.
For example, Microsoft has recently developed Graph-RAG \cite{graphrag2025}, a retrieval-augmented generation (RAG) method that enhances prompts by searching external knowledge to improve the accuracy and trustworthiness of LLM generation, and builds a hierarchical index by using the communities detected by Leiden.
%
%
As another example, Liu et al. introduced eRiskComm \cite{liu2022eriskcom}, a community-based fraud detection system that assists regulators in identifying high-risk individuals from social networks by using Louvain to partition communities, and Leiden can be naturally applied in this context. 


In the aforementioned application scenarios, the graphs often evolve frequently over time, with many insertions and deletions of vertices and edges.
We highlight two representative such tasks in \Acom.
(1) \emph{Anomaly detection:} This task's pipeline operates on the graph with $10^8$ vertices and $2\times 10^9$ edges, where recomputing Leiden communities from scratch introduces a delay of several hours.
(2) \emph{Social recommendation:} This task works on a graph with $10^8$ vertices and $4\times 10^{9}$ edges.
However, the online monitoring and operational workflows require incremental maintenance for batches of roughly $10^5$ edge updates within a few minutes.
Full recomputation of Leiden communities incurs hour-level delays, while product requirements demand minute-level updates for update batches of this scale.
%
Hence, it is strongly desirable to develop efficient algorithms for maintaining the up-to-date Leiden communities in large dynamic graphs.

{\bf Prior works.}
To maintain Louvain communities in dynamic graphs, several algorithms have been developed, such as DF-Louvain \cite{sahu2024df}, Delta-Screening \cite{zarayeneh2021delta}, DynaMo \cite{zhuang2019dynamo}, and Batch \cite{chong2013incremental}.
However, little attention has been paid to maintaining Leiden communities.
%
To the best of our knowledge, \cite{sahu2024starting} is the only work that achieves this.
It first uses some optimizations for the first iteration of {\tt DF-Leiden}, and then invokes the original Leiden algorithm for the remaining iterations, as depicted in Fig. \ref{fig:inc_methods}\subref{fig:inc_methods:existing}.
Following the optimized movement phase ({\tt opt-movement}), the refinement phase in DF-Leiden separates communities affected by edge or vertex changes into multiple sub-communities, while leaving unchanged communities as single sub-communities.
The aggregation phase remains identical to that of the Leiden algorithm.
After constructing the aggregated graph, the standard Leiden algorithm is applied to complete the remaining CD process.
The authors have also developed two variants of {\tt DF-Leiden}, called {\tt ND-Leiden} and {\tt DS-Leiden}, by using different optimizations for the movement phase of the first iteration.
Nevertheless, there is a lack of detailed theoretical analysis for these algorithms, and they are inefficient for large graphs with few changes.

\begin{figure}[t]
  \small
  \centering
  \subfloat[The process of the incremental algorithms in \cite{sahu2024starting}.]{
    \begin{minipage}[t]{0.48\textwidth}
      \centering
      \includegraphics[width=\linewidth]{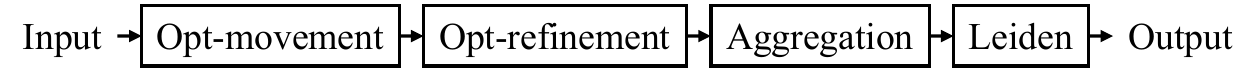}
    \end{minipage}
    \label{fig:inc_methods:existing}
  }
  \hfil
  \subfloat[The process of our \texttt{HIT-Leiden} algorithm.]{
    \begin{minipage}[t]{0.48\textwidth}
      \centering
      \includegraphics[width=\linewidth]{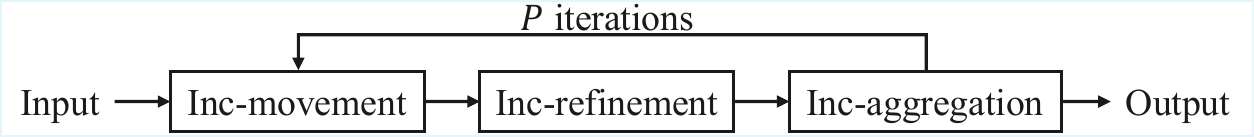}
    \end{minipage}
    \label{fig:inc_methods:hit}
  }
  \caption{Algorithms for maintaining Leiden communities.}
  \label{fig:inc_methods}
\end{figure}


\textbf{Our contributions.}
To address the above limitations, we first theoretically analyze the time cost of existing algorithms for maintaining Leiden communities and theoretically show that they are relatively unbounded via the boundedness analysis, which is a powerful tool for analyzing the time complexity of incremental algorithms on dynamic graphs.
We further analyze the membership of vertices in communities and sub-communities when the graph edges change, and observe that the procedure for maintaining these memberships generalizes naturally to all the supergraphs generated by Leiden.
The above analysis not only lays a solid foundation for us to comprehend existing algorithms but also offers us opportunities to improve upon them.

Based on the above analyses, we develop a novel efficient maintenance algorithm, called \underline{H}ierarchical \underline{I}ncremental \underline{T}ree Leiden (HIT-Leiden), which effectively reduces the range of affected vertices by maintaining the connected components and hierarchical community structures.
As depicted in Fig. \ref{fig:inc_methods}\subref{fig:inc_methods:hit}, {\tt HIT-Leiden} is an iterative algorithm with each iteration having three key phases, namely incremental movement, incremental refinement, and incremental aggregation, abbreviated as {\tt inc-movement}, {\tt inc-refinement}, and {\tt inc-aggregation}, respectively.
More specifically, {\tt inc-movement} extends the movement phase from \cite{sahu2024starting} by incorporating hierarchical community structures \cite{traag2019louvain}.
Unlike prior approaches, it operates on a supergraph where each supervertex represents a sub-community, focusing on hierarchical dependencies between communities and their nested substructures.
%
%
Inspired by the key technique of maintaining the connected components in dynamic graphs \cite{xu2024constant}, {\tt inc-refinement} maintains sub-communities by using tree-based structures to efficiently track changes in sub-communities.
{\tt Inc-aggregation} updates the supergraph by computing structural changes based on the outputs of the previous two phases.


We have evaluated {\tt HIT-Leiden} on several large-scale real-world dynamic graph datasets. 
The experimental results show that our algorithm achieves comparable community quality with the state-of-the-art algorithms for maintaining Leiden communities, while achieving speedups of up to five orders of magnitude over {\tt DF-Leiden}.
In addition, we have deployed our algorithm in real-world applications at \Acom. 

\textbf{Outline.}
We first review related work in Section \ref{sec:related}.
We then formally introduce some preliminaries, including the Leiden algorithm and problem definition in Section \ref{sec:pre}, provide some theoretical analysis in Section \ref{sec:thereom}, and present our proposed {\tt HIT-Leiden} algorithm in Section \ref{sec:hit}.
Finally, we present the experimental results in Section \ref{sec:exp} and conclude in Section \ref{sec:con}.
\section{Related Work}
\label{sec:related}

In this section, we first review the existing works of CD for both static and dynamic graphs.
We simply classify these works as modularity and other metrics-based CD methods.

$\bullet$ \textbf{Modularity-based CD.} 
Modularity-based CD methods aim to partition a graph such that communities exhibit high internal connectivity relative to a null model. 
Among these methods, Louvain \cite{blondel2008fast} is the most popular one due to its high efficiency and scalability, as shown in some comparative analyses \cite{yang2016comparative,amira2023survey}.
Leiden~\cite{traag2019louvain} improves upon Louvain by resolving the problem of disconnected communities, yielding higher-quality results with comparable runtime.
Other methods employ different modularity heuristics \cite{newman2006modularity}, incorporate simulated annealing \cite{kirkpatrick1983optimization}, spectral techniques \cite{newman2013spectral}, or evolutionary strategies \cite{liu2019evolutionary}.
%
%
Recent neural approaches have integrated modularity objectives into deep learning models \cite{zhang2020commdgi,yang2016modularity,xie2018community}, enhancing representation learning for CD.

Besides, some recent works have studied how to incrementally maintain modularity-based communities when the graph is changed.
Aynaud et al. \cite{aynaud2010static} proposed one of the earliest approaches by reusing previous community assignments to warm-start the Louvain algorithm.
Subsequent works extended this idea to both Louvain \cite{zhuang2019dynamo,zarayeneh2021delta,sahu2024df} and Leiden \cite{sahu2024starting}, incorporating mechanisms such as edge-based impact screening or localized modularity updates.
Nevertheless, the existing algorithms for maintaining Leiden communities lack in-depth theoretical analysis, and their practical efficiency is poor.
Other methods based on modularity, including extensions to spectral clustering \cite{chi2009evolutionary}, multi-step CD \cite{aynaud2011multi}, and label propagation-based methods \cite{xie2013labelrankt} have been studied on dynamic graphs.

$\bullet$ \textbf{Other metrics-based CD.}
Beyond modularity, various CD methods have been developed for different optimization purposes, such as similarity, statistical inference, spectral clustering, and neural networks.
The similarity-based methods like SCAN \cite{ester1996density,xu2007scan, wang2025searching} identify dense regions from the graph via structural similarity.
Statistical inference approaches, including stochastic block models \cite{holland1983stochastic, karrer2011stochastic, airoldi2008mixed}, infer communities by fitting generative probabilistic models to observed networks.
Spectral clustering methods \cite{newman2006finding, de2014laplacian} exploit the eigenstructure of graph Laplacians to group nodes with similar structural roles.
Deep learning-based methods for CD have recently gained traction.
Graph convolutional networks \cite{hu2020going, cui2020adaptive, zhao2021graph}, and graph attention networks \cite{fu2020magnn, velivckovic2017graph, wang2021self} have demonstrated strong performance in learning expressive node embeddings for CD tasks.
For more details, please refer to recent survey papers of CD \cite{chakraborty2017metrics,su2022comprehensive}.

Besides, many of the above methods have also been extended for dynamic graphs.
Ruan et al. \cite{ruan2021dynamic} and Zhang et al. \cite{zhang2022effective} have studied structural graph clustering on dynamic graphs, which is based on structural similarity.
Temporal spectral methods \cite{chi2007evolutionary} and dynamic stochastic block models \cite{matias2017statistical} enable statistical modeling of evolving community structures over time.
Recent deep learning approaches also support dynamic CD through mechanisms such as dynamic embeddings \cite{yang2017graph} and contrastive learning \cite{park2022cgc}.
These models capture temporal dependencies and structural evolution.

\section{Preliminaries}
\label{sec:pre}

In this section, we first briefly recap the original Leiden procedure, then briefly introduce the maintenance of Leiden communities problem, and finally introduce the key notations.

\subsection{Leiden algorithm}
\label{sec:problem}

\begin{algorithm}[t]
  \KwIn{$G$, $f(\cdot)$, $P$, $\gamma$}
  
  \KwOut{Updated $f(\cdot)$}
  
  $G^1 \leftarrow G$, $f^1(\cdot) \leftarrow f(\cdot)$\;
  
  \For{$p = 1$ to $P$}{
    $f^{p}(\cdot) \leftarrow Move(G^p, f^p(\cdot), \gamma)$\;
    $s^{p}(\cdot) \leftarrow Refine(G^p, f^p(\cdot), \gamma)$\;
    \If{$p < P$}{
      
      
      $G^{p + 1}, f^{p + 1}(\cdot) \leftarrow Aggregate(G^p, f^{p}(\cdot), s^p(\cdot))$\;
    }
  }
  
  Update $f(\cdot)$ using $s^1(\cdot), \cdots, s^{P}(\cdot)$\;
  
  \KwRet $f(\cdot)$\;

  \caption{Leiden algorithm \cite{platoRepo, sahu2024fast}}
  \label{alg:leiden}
\end{algorithm}

Originally, the Leiden algorithm was designed for static graphs.
It mainly consists of three phases: {\it movement}, {\it refinement}, and {\it aggregation}, which are executed iteratively.
Algorithm~\ref{alg:leiden} \cite{platoRepo, sahu2024fast} presents the Leiden algorithm following the process in Fig.~\ref{fig:louvain_leiden}\subref{fig:louvain_leiden:leiden}. 
Starting from the original graph $G^1=G$ and an initial community mapping $f(\cdot)$, Leiden iterates for $P$ rounds (lines 1-2) where $P$ is a pre-defined parameter that users often set in practice to balance modularity gains against runtime.
$\gamma$ is a hyperparameter to control the scale of communities \cite{reichardt2006statistical}.
In each round $p$, the \textit{movement} phase relocates vertices (or supervertices) to improve modularity, producing an updated mapping $f^p(\cdot)$ from the vertices (or supervertices) in this iteration to communities (line 3).
The \textit{refinement} phase splits each community into connected sub-communities, producing a sub-community mapping $s^p(\cdot)$ from the supervertices in this iteration to the supervertices in the next iteration (line 4).
The \textit{aggregation} phase contracts each sub-community into a supervertex to form the next-level supergraph $G^{p+1}$ (lines 5-6).
After $P$ rounds, the final partition $f(\cdot)$ of $G$ is obtained by composing the mappings across levels.

The overall time complexity of Algorithm~\ref{alg:leiden} is $O\left(P\cdot(|V| + |E|\right))$, since each iteration costs $O(|V| + |E|)$ time \cite{sahu2024fast}, where $V$ and $E$ are the vertex set and edge set of $G$, respectively.

\begin{figure}[t]
  \centering
  \subfloat[All the communities.]{
    \begin{minipage}[t]{0.199938194109126\textwidth}
      \centering
      \includegraphics[width=\linewidth]{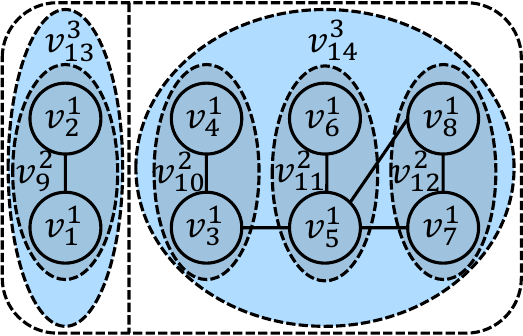}
    \end{minipage}
    \label{fig:hie_structural:graph}
  }
  \hfil
  \subfloat[A tree-like structure.]{
    \begin{minipage}[t]{0.236\textwidth}
      \centering
      \includegraphics[width=\linewidth]{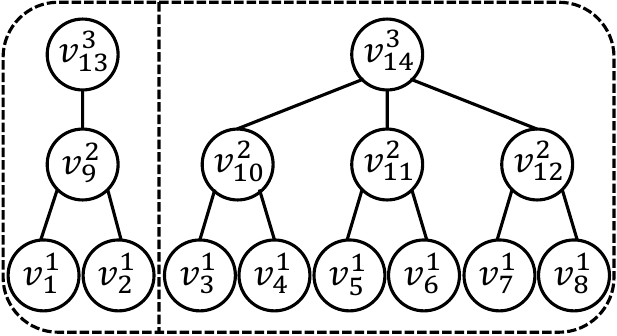}
    \end{minipage}
    \label{fig:hie_structural:hie_structural}
  }
  \caption{The process of Leiden for the graph $G$ in Fig. \ref{fig:mod}\protect\subref{fig:mod:org_graph}.}
  \label{fig:hie_structural}
\end{figure}

\begin{figure*}[t]
  \small
  \centering
  \subfloat[The process of hierarchical partitions at the first iteration on the graph.]{
    \begin{minipage}[t]{0.771047430830039\textwidth}
      \centering
      \includegraphics[width=\textwidth]{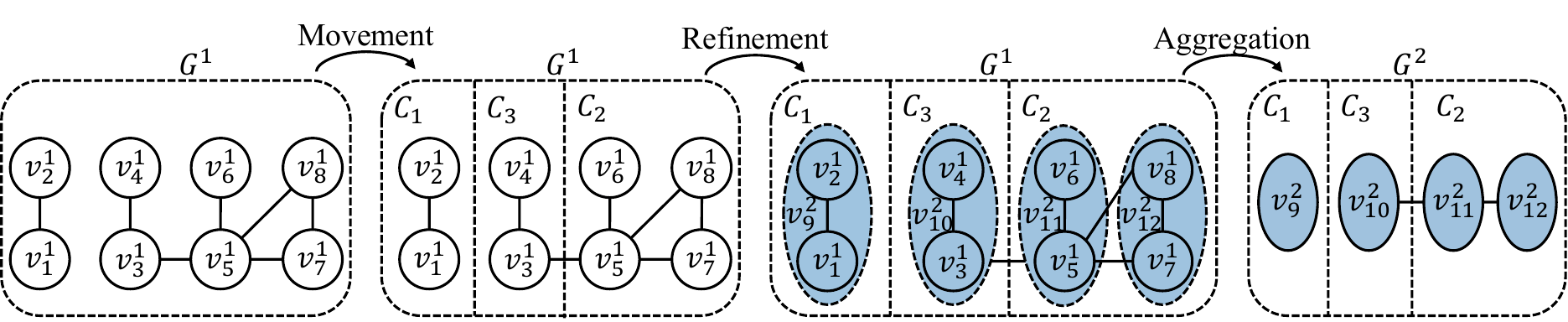}
    \end{minipage}
    \label{fig:l1_pro:graph}
  }
  \hfil
  \subfloat[The tree-like structure.]{
    \begin{minipage}[t]{0.2\textwidth}
      \centering
      \includegraphics[width=0.818181818181818\textwidth]{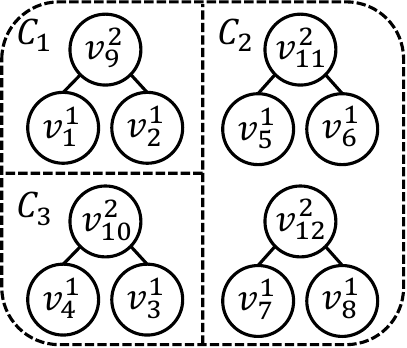}
    \end{minipage}
    \label{fig:l1_pro:hie_structural}
  }
  \caption{The process of hierarchical partitions of Fig. \ref{fig:hie_structural} at level-1 with the Leiden algorithm.}
  \label{fig:l1_pro}
\end{figure*}

\begin{example}
\label{eg:leiden}
Fig. \ref{fig:hie_structural}\subref{fig:hie_structural:graph} depicts the process of Leiden with $P$=3 for the graph in Fig. \ref{fig:mod}\subref{fig:mod:org_graph}.
Let $v_i^p$ denote the supervertex in the $p$-th iteration of Leiden.
It generates three levels of supergraphs: $G^1$, $G^2$, and $G^3$, with $G^1 = G$.
The vertices of these supergraphs form a tree-like structure, as shown in Fig. \ref{fig:hie_structural}\subref{fig:hie_structural:hie_structural}. 

Take the first iteration as an example depicted in Fig. \ref{fig:l1_pro}.
In the movement phase, it generates three communities $C_1 = \{v^1_1, v^1_2\}$, $C_2 =\{ v^1_5, v^1_6, v^1_7, v^1_8\}$ and $C_3 = \{v^1_3, v^1_4\}$.
In the refinement phase, $C_2$ is split into two sub-communities $v^2_{11}=\{v^1_5, v^1_6\}$ and $v^2_{12}=\{v^1_7, v^1_8\}$, and $C_1$ and $C_3$ are unchanged.
In the aggregation phase, all vertices are aggregated into supervertices based on their sub-community memberships, resulting in $G^2$.
\end{example}

\subsection{Maintaining Leiden communities in dynamic graphs}
\label{sec:dynamic_overview}

For simplicity, we call the communities generated by the Leiden algorithm {\it Leiden communities}.
The goal of maintaining Leiden communities is to refresh the communities after each graph update.
We now formally introduce the problem studied in this paper.

\begin{problem}[Maintenance of Leiden communities \cite{sahu2024starting}]
\label{prob:ourProblem}
Given a graph $G$ with its Leiden communities $\mathbb{C}$, and some edge updates $\Delta G$, return the updated Leiden communities after applying $\Delta G$ to $G$.
\end{problem}

As shown in Algorithm~\ref{alg:leiden}, the Leiden algorithm has three phases.
As a result, to maintain Leiden communities when the graph is updated, we can follow the same three-phase procedure, but only make updates in a small affected subgraph:
{\it (1)  Movement phase:} we update the community assignments for affected vertices to improve modularity;
{\it (2) Refinement phase:} we restore the connectivity guarantees by splitting and merging affected parts inside each community;
{\it (3) Aggregation phase:} we update the supergraph and propagate structural changes to higher levels when needed.

However, this task is challenging for two reasons.
First, the refinement phase enforces intra-community connectivity, which requires efficiently tracking connectivity changes under edge insertions or deletions.
Second, because the refinement and aggregation phases are tightly coupled, even a small local update may cascade across multiple hierarchical levels, leading to unstable and potentially large maintenance costs.

\subsection{Notation definition}
\label{subsec:notation}

We consider an \textbf{undirected and weighted graph} $G = (V, E)$, where $V$ and $E$ are the sets of vertices and edges, respectively.
Each vertex $v$'s neighbor set is denoted by $N(v)$.
Each edge $(v_i, v_j)$ is associated with a positive weight $w(v_i, v_j) > 0$.
The degree of $v_i$ is given by $d(v_i) = \sum_{v_j \in N(v_i)} w(v_i, v_j)$.
Denote by $m$ the total weight of all edges in $G$, i.e., $m = \sum_{(v_i, v_j) \in E} w(v_i, v_j)$.
For convenience, Table~\ref{tab:notations_and_meaning} provides a quick reference for frequently used symbols.

Given a graph $G=(V, E)$, the CD process aims to partition all the vertices of $V$ into some disjoint sets $\mathbb{C}$, each of which is called a community, corresponding to a set of vertices that are densely connected.
This process can be modeled as a mapping function $f(\cdot): V \rightarrow \mathbb{C}$, such that each $v$ belongs to a community $f(v)$ of the partition $\mathbb{C}$.
For each vertex $v$, the total weight between $v$ and a community $C$ is denoted by $w(v, C) = \sum_{v' \in N(v) \cap C} w(v, v')$.
%

\begin{table}
  \caption{Frequently used notations and their meanings.}
  \label{tab:notations_and_meaning}
  \begin{tabular}{ M{1.98cm} | L{6.02cm} }
      \hline
      \textbf{Notation} & \textbf{Meaning}\\
      \hline\hline
      $G$ = $(V,E)$               & A graph with vertex set $V$ and edge set $E$ \\
      \hline
      $N(v)$, $N_2(v)$            &The vertex $v$'s 1- and 2-hop neighbor sets, resp.\\
      \hline
      $w(v_i, v_j)$               & The weight of edge between $v_i$ and $v_j$ \\
      \hline
      $d(v)$                      & The weighted degree of vertex $v$\\
      \hline     
      $m$                         & The total weight of all edges in $G$\\
      \hline

      $\mathbb{C}$                &A set of communities forming a partition of $G$ \\
      \hline        
      $Q$                         &The modularity of the graph $G$ with partition $\mathbb{C}$ \\ 
      \hline
      $G^p$=$(V^p, E^p)$          &The supergraph in the $p$-th iteration of Leiden\\
      \hline
      $\Delta Q(v \to C',\gamma)$  & Modularity gain by moving $v$ from $C$ to $C'$ with $\gamma$\\
      \hline
      $f(\cdot)$: $V \to \mathbb{C}$           &A mapping from vertices to communities\\
      \hline        
      $f^p(\cdot)$: $V^p \to \mathbb{C}$       &A mapping from supervertices to communities\\
      \hline        
      $s^p(\cdot)$: $V^p \to V^{p+1}$          &A mapping from supervertices in $p$-th level to supervertices in $(p+1)$-th level (sub-communities)\\
      \hline        
      $V^p$, $E^p$          & A set of supervertices and superedges in $p$-th level\\
      \hline       
      $\Delta G$                  &The set of changed edges in the dynamic graph\\
      \hline
  \end{tabular}
\end{table}

As a well-known CD metric, the modularity measures the difference between the actual number of edges in a community and the expected number of such edges.

\begin{definition}[Modularity~\cite{blondel2008fast}]
\label{def:modularity}
Given a graph $G = (V, E)$ and a community partition $\mathbb{C}$ over $V$, the modularity $Q(G,\mathbb{C},\gamma)$ of the graph $G$ with the partition $\mathbb{C}$ is defined as:
\begin{equation}
      \small
      Q(G,\mathbb{C},\gamma) =  \sum_{C \in \mathbb{C}} \left( \frac{1}{2m} \sum_{v\in C}w(v, C) - \gamma\left(\frac{d(C)}{2m}\right)^2  \right),
      \label{equ:modularity}
\end{equation}
where $d(C)$ is the total degree of all vertices in a community $C$, and $\gamma > 0$ is a hyperparameter.
\end{definition}

Note that the parameter $\gamma > 0$ controls the granularity of the detected communities \cite{reichardt2006statistical}. 
A higher $\gamma$ favors smaller, finer-grained communities.
In practice, $\gamma$ is often set to 0.5, 1, 4, or 32, as shown in \cite{lindeboom2024human}.
Besides, to guide community updates, the concept of modularity gain is often used to capture the changed modularity when a vertex is moved from one community to another.

\begin{definition}[Modularity gain \cite{blondel2008fast}]
\label{def:modularity_gain}
Given a graph $G$, a partition $\mathbb{C}$, and a vertex $v$ that belongs to a community $C$, the modularity gain of moving $v$ from $C$ to another community $C'$ is defined as:
\begin{equation}
      \small
      \begin{aligned}
        \Delta Q(v \to C',\gamma) = &\frac{w(v, C') - w(v, C)}{m} \\
        &+ \frac{\gamma \cdot d(v) \cdot \left( d(C) - d(v) - d(C') \right)}{2 m^2}.
      \end{aligned}
\label{equ:gain}
\end{equation}
\end{definition}

In this paper, we focus on the dynamic graph with insertions and deletions of both vertices and edges.
Since a vertex insertion (resp. deletion) can be modeled as a sequence of edge insertions  (resp. deletions), we simply focus on edge changes.
Given a set of edge changes $\Delta G$ to a graph $G=(V, E)$, we obtain an updated graph $G' = (V', E')$.
Since there are two types of edge updates, we let $\Delta G = \Delta G_+ \cup \Delta G_-$, where $\Delta G_+ = E' \setminus E$ and $\Delta G_- = E \setminus E'$ denote the sets of inserted and deleted edges, respectively.
We denote updated edges $(v_i, v_j, \alpha) \in \Delta G_+$ and $(v_i, v_j, -\alpha) \in \Delta G_-$, where $\alpha$ is positive, i.e., $\alpha > 0$. 
We use $G \oplus \Delta G$ to denote applying $\Delta G$ to $G$, yielding an updated graph $G'$.

\section{Theoretical Analysis of Leiden}
\label{sec:thereom}

In this section, we first analyze the boundedness of existing algorithms, then study how vertex behavior impacts community structure under graph updates, and extend our analysis to supergraphs.
For lack of space, all the proofs of lemmas are shown in the appendix of the full version \cite{hitleiden} of this paper.

\subsection{Boundedness analysis}
\label{sec:boundedness}

We first introduce some concepts related to boundedness.

\textbf{$\bullet$ Notation.}
Let $\Theta$ denote the CD query applied to a graph $G$, where $\Theta(G)=\mathbb{C}$ is the set of detected communities. 
The new graph is $G \oplus \Delta G$, and the updated community is $\Theta(G \oplus \Delta G)$. 
We denote the output difference as $\Delta \mathbb{C}$, where $\Theta(G \oplus \Delta G) = \Theta(G) \oplus \Delta \mathbb{C}$.                             

\textbf{$\bullet$ Concepts of boundedness.} The notion of boundedness \cite{ramalingam1996computational} evaluates the effectiveness of an incremental algorithm using the metric ${\tt CHANGED}$, defined as ${\tt CHANGED} = \Delta G + \Delta \mathbb{C}$, which leads to $|{\tt CHANGED}| = |\Delta G| + |\Delta \mathbb{C}|$.

\begin{definition}[Boundedness \cite{ramalingam1996computational,fan2017incremental}]
  An incremental algorithm is bounded if its computational cost can be expressed as a polynomial function of $|{\tt CHANGED}|$ and $|\Theta|$. Otherwise, it is unbounded.
\end{definition}

\textbf{$\bullet$ Concepts of relative boundedness.} In real-world dynamic graphs, $|{\tt CHANGED}|$ is often small, yet some unbounded algorithms can be solved in polynomial time using measures comparable to $|{\tt CHANGED}|$, making these algorithms feasible. 
To assess these incremental algorithms effectively, Fan et al. \cite{fan2017incremental} introduced the concept of relative boundedness, which leverages a more refined cost model called the affected region.
Let ${\tt AFF}$ denote the affected part, the region of the graph actually processed by the incremental algorithm. 

\begin{definition}[${\tt AFF}$ \cite{fan2017incremental}]
\label{lem:AFF}
  Given a graph $G$, a query $\Theta$, and the input update $\Delta G$ to $G$, ${\tt AFF}$ signifies the cost difference of the static algorithm between computing $\Theta(G)$ and $\Theta(G \oplus \Delta G)$.
\end{definition}

Unlike ${\tt CHANGED}$, ${\tt AFF}$ captures the concrete portion of the graph touched by an incremental algorithm, providing a tighter bound on its computational cost. This leads to the following definition.

\begin{definition}[Relative boundedness \cite{fan2017incremental}]
\label{lem:rel_bound}
  An incremental graph algorithm is relatively bounded to the static algorithm if its cost is polynomial in $|\Theta|$ and $|{\tt AFF}|$.
\end{definition}

We now analyze the boundedness of existing incremental Leiden algorithms.


\begin{theorem}
\label{theo:unbounded}
When processing an edge deletion or insertion, the incremental Leiden algorithms proposed in \cite{sahu2024starting} all cost $O\left(\sum_{p=1}^P \left(|V^p| + |E^p|\right)\right)$.
\end{theorem}

\begin{table}[t]
  \caption{Incremental Leiden algorithms}
  \label{tab:related_works}
  \begin{tabular}{ c | M{4cm} | M{1.4cm} }
      \hline
      \textbf{Method} & \textbf{Time complexity} & \shortstack[c]{\rule{0pt}{1.1em}\textbf{Relative} \\[-2pt] \textbf{boundedness}} \\
      \hline\hline
      {\tt ND-Leiden} \cite{sahu2024starting}     &   $O\left(\sum_{p=1}^P \left(|V^p| + |E^p|\right)\right)$    & \ding{55} \\
      {\tt DS-Leiden} \cite{sahu2024starting}     &   $O\left(\sum_{p=1}^P \left(|V^p| + |E^p|\right)\right)$    & \ding{55} \\
      {\tt DF-Leiden} \cite{sahu2024starting}     &   $O\left(\sum_{p=1}^P \left(|V^p| + |E^p|\right)\right)$    & \ding{55} \\
      \hline
      \textbf{\texttt{HIT-Leiden}}    &   $O\left(|N_2(\texttt{CHANGED})| + |N_2(\texttt{AFF})|\right)$    & \ding{51} \\
      \hline        
  \end{tabular}
\end{table}

By Theorem \ref{theo:unbounded}, the existing algorithms for maintaining Leiden communities are both unbounded and relatively unbounded, as shown in Table \ref{tab:related_works}.
They are very costly for large graphs, even for small updates.
Next, we review the properties of Leiden and then identify {\tt AFF} of Leiden.
\subsection{Vertex optimality and subpartition $\gamma$-density}
\label{subsec:vectex_opt}
As shown in the literature \cite{blondel2008fast,traag2019louvain}, if $s^P(\cdot)=f^P(\cdot)$ after $P$ iterations, Leiden is guaranteed to satisfy the following two properties:
\begin{itemize}
    \item {\bf Vertex optimality:} All the vertices are vertex optimal.
    
    \item {\bf Subpartition $\gamma$-density:} All the communities are subpartition $\gamma$-dense.
\end{itemize}
%

To design an efficient and effective maintenance algorithm for Leiden communities, we analyze the behaviors of vertices and communities when the graph changes as follows.

$\bullet$ {\bf Analysis of vertex optimality.}
We begin with a key concept.

\begin{definition}[Vertex optimality \cite{blondel2008fast}]
\label{def:vertex-optimality}
A community $C \in \mathbb{C}$ is called vertex-optimal if for each vertex $v \in C$ and $C' \in \mathbb{C}$, the modularity gain $\Delta Q(v \to C', \gamma) \leq 0$.
\end{definition}

Next, we introduce an assumption in the maintenance of Louvain communities \cite{zarayeneh2021delta,sahu2024df}:
\begin{assumption}
\label{ass:size_diff}
    The sum of weights of the updated edges is sufficiently small relative to the graph size $m$. 
    %
\end{assumption}

Based on Assumption \ref{ass:size_diff}, prior studies suggest that when the number of edge updates is small relative to the graph size, three heuristics hold:
(1) intra-community edge deletions and inter-community edge insertions could affect vertex-level community membership~\cite{zarayeneh2021delta,sahu2024df};
(2) Inter-community edge deletions and intra-community edge insertions can be ignored \cite{zarayeneh2021delta,sahu2024df};
(3) Vertices directly involved in such edge changes are the most likely to alter their communities~\cite{sahu2024df}.
The heuristics are stated in Observation \ref{lemma:intra-cross-edge}, which can be proved based on Definition \ref{def:vertex-optimality}.

\begin{observation} [\cite{sahu2024df}]
\label{lemma:intra-cross-edge}
Given an intra-community edge deletion $(v_i, v_j, -\alpha)$ or a cross-community edge insertion $(v_i, v_j, \alpha)$, its effect on the community memberships of vertices $v_i$ and $v_j$ cannot be ignored.
\end{observation}

We further derive the propagation of community changes from Observation \ref{lemma:intra-cross-edge}.

\begin{lemma}
\label{lemma:neighbor-change}
When a vertex $v$ changes its community to $C$, then the communities of its neighbors not in $C$ in the updated graph could be affected.
\end{lemma}


Based on these analyses, we develop a novel movement phase, called {\tt inc-movement} in {\tt HIT-Leiden} to preserve vertex optimality, which will be introduced in Section \ref{subs:hifm}.

$\bullet$ {\bf Analysis of subpartition $\gamma$-density.} 
For simplified analysis, we first introduce $\gamma$-order and $\gamma$-connectivity, which are key concepts for defining subpartition $\gamma$-density.

\begin{definition}[$\gamma$-order]
\label{def:go}
%
Given two vertex sequences $X$ and $Y$ of a graph $G$, let $X \otimes Y$ represent that $Y$ is merged into $X$ such that $2m \cdot w(X, Y) \geq \gamma \cdot d(X)\cdot d(Y)$, where $w(X, Y) = \sum_{v_i \in X} \sum_{v_j \in Y} w(v_i, v_j)$.
A $\gamma$-order of a vertex sequence $U = \{v_1, \cdots, v_x\}$ represents the merged sequence starting from singleton sequences $\{v_1\}, \cdots, \{v_x\}$. 
\end{definition}

We can maintain one $\gamma$-order per sub-community from {\tt Leiden}, which is represented by the sequence of vertices merging into the sub-community in {\tt refinement} phase of {\tt Leiden}.

\begin{definition}[$\gamma$-connectivity \cite{traag2019louvain}]
\label{def:gcc}
Given a graph $G$, a vertex sequence $U$ is $\gamma$-connected if $U$ can be generated from at least one $\gamma$-order. 
\end{definition}



\begin{definition}[Subpartition $\gamma$-density \cite{traag2019louvain}]
\label{def:sgd}
A vertex subsequence $U \subseteq C \in \mathbb{C}$ is subpartition $\gamma$-dense if $U$ is $\gamma$-connected, and any intermediate vertex sequence $X$ is locally optimized, i.e., $\Delta Q(X \to \emptyset,\gamma) \leq 0$.
\end{definition}

Notably, $\Delta Q(X \to \emptyset,\gamma) \leq 0$ denotes the modularity gain of moving $X$ from $C$ to an empty set, whose calculation follows the same formula as the standard modularity gain in \eqref{equ:gain}.


\begin{figure}[t]
  \centering
  \subfloat[A triangle.]{
    \begin{minipage}[t]{0.15\textwidth}
      \centering
      \includegraphics[width=0.537801333333334\textwidth]{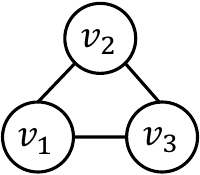}
    \end{minipage}
    \label{fig:triangle:org}
  }
  \subfloat[Delete an edge.]{
    \begin{minipage}[t]{0.15\textwidth}
      \centering
      \includegraphics[width=0.537801333333334\textwidth]{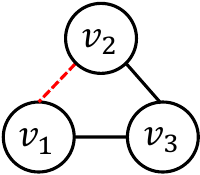}
    \end{minipage}
    \label{fig:triangle:cur}
  }
  \subfloat[Delete two edges.]{
    \begin{minipage}[t]{0.15\textwidth}
      \centering
      \includegraphics[width=0.537801333333334\textwidth]{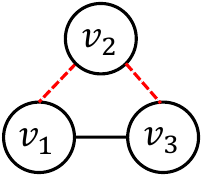}
    \end{minipage}
    \label{fig:triangle:cur2}
  }
  \caption{An example illustrating subpartition $\gamma$-density.}
  \label{fig:triangle}
\end{figure}

\begin{example}
\label{eg:go}
The triangle in Fig. \ref{fig:triangle}\subref{fig:triangle:org} is subpartition $\gamma$-dense with $\gamma = 1$ since there are six different $\gamma$-orders.
For instance, one is $\{v_3\} \otimes (\{v_1\} \otimes \{v_2\})$, which represents that ${v_2}$ is merged into $\{v_1\}$ generating sequence $\{v_1, v_2\}$, and then $\{v_1, v_2\}$ is merged into ${v_3}$ generating $\{v_1, v_2, v_3\}$.
After deleting the edge $(v_1, v_2)$, although $\{v_3\} \otimes (\{v_1\} \otimes \{v_2\})$ is not a $\gamma$-order, the updated graph is still subpartition $\gamma$-dense since $\{v_1\} \otimes (\{v_2\} \otimes \{v_3\})$ is a $\gamma$-order in the updated graph.
After continuing to delete the edge $(v_2, v_3)$, the updated graph is not subpartition $\gamma$-dense since $v_2$ is not connected to $v_1$ and $v_3$.
\end{example}

In essence, each community $C$ (or sub-community $S$) of Leiden is subpartition $\gamma$-dense, since (1) any sub-community in $C$ (or $S$) is locally optimized, and (2) all sub-communities are $\gamma$-connected. 
Notably, as shown in Fig. \ref{fig:inc_methods}\subref{fig:inc_methods:hit}, vertex optimality ensures the first condition by design since any sub-community will be a supervertex in {\tt inc-movement} of the next iteration.
Thus, we will develop a new refinement algorithm, {\tt inc-refinement}, to preserve $\gamma$-connectivity of sub-communities.

Next, we analyze the $\gamma$-connectivity property under two kinds of graph updates, i.e., \emph{edge deletion} and \emph{edge insertion}.
%
%
For any vertex $v_i$ within a sub-community $S_i$ with a $\gamma$-order, we denote an intermediate subsequence of the $\gamma$-order containing $v_i$ by $I_i \subseteq S_i$, and the subsequence $U_i = I_i \setminus \{v_i\}$ is an intermediate subsequence of the $\gamma$-order before merging $v_i$.
%

\textit{\underline{(1) Edge deletion.}}
We consider the deletions of both intra-sub-community edges and cross-sub-community edges:
\begin{lemma}
    Given an intra-sub-community edge deletion $(v_i, v_j, -\alpha)$, assume $v_j$ is before $v_i$ in the $\gamma$-order of the sub-community.
    The effects of the edge deletion can be described by the following four cases:
    \begin{enumerate}
        \item $v_i$ could be removed from its sub-community only if $\alpha > \frac{2m \cdot w(v_i, U_i) - \gamma \cdot d(v_i) \cdot d(U_i)}{4m + 2w(v_i, U_i)}$;
        
        \item $v_j$ could be removed from its sub-community only if $\alpha > m - \frac{ \gamma \cdot d(v_j) \cdot d(U_j)}{2w(v_j, U_j)}$;
        
        \item For any $v_k \in S_i$ ($k \neq i, j$), it could be removed from its sub-community only if $\alpha > m - \frac{\gamma \cdot d(v_k) \cdot d(U_k)}{2w(v_k, U_k)}$;
        
        \item For any $v_l \notin S_i$, it should be removed from its sub-community if and only if $\alpha > m - \frac{\gamma \cdot d(v_l) \cdot d(U_l)}{2w(v_l, U_l)}$.
    \end{enumerate}
    \label{lem:intra_sub_e_del}
\end{lemma}

\begin{lemma}
    Given a cross-sub-community edge deletion $(v_i, v_j, -\alpha)$, the effects of the edge deletion can be described by the following four cases:
    \begin{enumerate}
        \item $v_i$ could be removed from its sub-community only if $\alpha > m - \frac{\gamma \cdot d(v_i) \cdot d(U_i)}{2w(v_i, U_i)}$;
        
        \item $v_j$ holds similar behavior with $v_i$;
        
        \item For any $v_k \in S_i \cup S_j$ ($k \neq i, j$), it could be removed from its sub-community only if $\alpha > m - \frac{\gamma \cdot d(v_k) \cdot d(U_k)}{2w(v_k, U_k)}$;
        
        \item For any $v_l \notin S_i \cup S_j$, it could be removed from its sub-community only if $\alpha > m - \frac{\gamma \cdot d(v_l) \cdot d(U_l)}{2w(v_l, U_l)}$.
    \end{enumerate}
    \label{lem:cross_sub_e_del}
\end{lemma}

\textit{\underline{(2) Edge insertion.}}
%
We consider the insertion of edges containing the insertions of both intra-sub-community edges and cross-sub-community edges:
\begin{lemma}
  Given an edge insertion $(v_i, v_j, \alpha)$, the effects of the edge insertion can be described by the following four cases:
    \begin{enumerate}
        \item $v_i$ could be removed from its sub-community only if $\alpha > \frac{4}{\gamma}m - d(I_i)$ or $\alpha > \frac{2w(v_i, U_i)}{\gamma \cdot d(U_i)} \cdot m - d(v_i)$;
        
        \item  $v_j$ could be removed from its sub-community, only if $\alpha > \frac{2w(v_j, U_j)}{\gamma \cdot d(U_j)} \cdot m - d(v_j)$;
        
        \item For any $v_k \in S_i \cup S_j$ ($k \neq i, j$), it could be removed from its sub-community, only if $\alpha > \frac{w(v_k, U_k)}{\gamma \cdot d(v_k)} \cdot m - \frac{1}{2}d(U_k)$;


        \item For any $v_l \notin S_i \cup S_j$, it is unaffected.
        
    \end{enumerate}
  \label{lem:sub_e_ins}
\end{lemma}

\begin{observation}
\label{obs:ord}
    In the refinement phase of Leiden algorithms, each vertex $v$ tends to be merged into the sub-community (intermediate subsequence $U$), offering more edge weights $w(v, U)$ and smaller degrees $d(U)$.
    Therefore, the differences in the values of $d(v)$, $w(v, U)$, and $d(U)$ are typically very small.
\end{observation}

Therefore, $\alpha$ is unlikely to satisfy the conditions in cases (2)-(4) of Lemma \ref{lem:intra_sub_e_del}, all the cases of Lemma \ref{lem:cross_sub_e_del}, and the conditions in cases (1)-(3) of Lemma \ref{lem:sub_e_ins} when $\alpha \ll m$ (which is often true as stated in Assumption \ref{ass:size_diff}).
As a result, when designing the maintenance algorithm, we only need to consider the effect of intra-sub-community edge deletions on $v_i$, which cannot be ignored.

Besides, our experiments show the following observation, which shows that the case (1) of Lemma 2 can also be ignored.

\begin{observation}
\label{obs:conn_compoent}
    Given an updated graph with its previous sub-community memberships, for any sub-community $S$, we treat each connected component in $S$ as a new sub-community. 
    Most of the maintained communities are subpartition $\gamma$-dense.
\end{observation}

The above observation holds because Leiden only offers us a $\gamma$-order from the refinement phase, and a subgraph often exists with multiple distinct $\gamma$-orders as shown in Example \ref{eg:go}.
Besides, if a vertex is a candidate affecting $\gamma$-connectivity, it is often a candidate affecting vertex optimality, e.g., the vertex $v_2$ in Fig. \ref{fig:triangle}\subref{fig:triangle:cur2}.
In this case, the vertex is likely to change its community before verifying whether the vertex needs to move out of its sub-community. 
Hence, the case (1) of Lemma \ref{lem:intra_sub_e_del} can be ignored if the intra-sub-community edge deletion does not cause the sub-community to be disconnected.

Based on Observations \ref{obs:ord}-\ref{obs:conn_compoent}, we develop a novel refinement algorithm, called {\tt inc-refinement}, in {\tt HIT-Leiden}, which will be introduced in Section \ref{subs:ticr}.

As shown in Fig. \ref{fig:gamma-desnity}, over $99\%$ maintained communities from {\tt HIT-Leiden} are subpartition $\gamma$-dense.

{\bf Extension to supergraphs.} Changes at the lower level propagate upward to superedge changes in the higher-level supergraph, as Leiden constructs a list of supergraphs in a bottom-up manner.
This motivates us to develop an incremental aggregation phase, namely {\tt inc-aggregation}, to compute the superedge changes in Section \ref{subs:da}.

\begin{example}
\label{eg:extend}
  In Fig. \ref{fig:mod}, communities $C_1$ and $C_2$ are treated as supervertices.
  Deleting an edge $(v_3, v_5, 1)$ and inserting an edge $(v_1, v_3, 1)$ cause $v_3$ and $v_4$ to move from $C_2$ to $C_1$.
  This results in the deletion of ${(C_2, C_2, -2)}$ and insertion of ${(C_1, C_1, 2)}$ in the supergraph.
\end{example}

Therefore, we treat each supergraph as a set of facing edge changes from the previous Leiden community and process them using a consistent procedure as shown in Fig. \ref{fig:inc_methods}\subref{fig:inc_methods:hit}.

\textbf{Characterization of \texttt{AFF}.} 
Based on these analyses, we define the supervertices that change their communities or sub-communities as the affected area {\tt AFF} of Leiden.

\section{Our HIT-Leiden algorithm}
\label{sec:hit}

In this section, we propose \texttt{HIT-Leiden}, which updates community memberships by restricting computation to the affected region, as motivated by the analysis in Section \ref{subsec:vectex_opt}.
We first introduce its three key components, namely {\tt inc-movement}, {\tt inc-refinement}, and {\tt inc-aggregation} of our {\tt HIT-Leiden}.
Fig. \ref{fig:theo2alg} shows the assumption, lemmas, and observations used in these components.
Then, we present an auxiliary procedure, called deferred update, abbreviated as {\tt def-update}.
Afterward, we give an overview of {\tt HIT-Leiden}, and finally analyze the boundedness of {\tt HIT-Leiden}.

\begin{figure}[t]
    \centering
    \includegraphics[width=\linewidth]{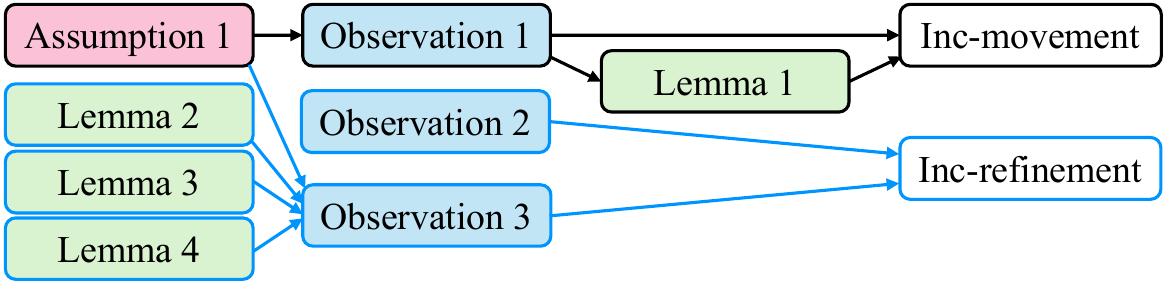}
    \setlength{\abovecaptionskip}{-17pt}
    \caption{The design rationale for \texttt{inc-movement} and \texttt{inc-refinement}.}
    \label{fig:theo2alg}
\end{figure}

\subsection{Inc-movement}
\label{subs:hifm}

The goal of {\tt inc-movement} is to preserve vertex optimality.
As analyzed in Section \ref{subsec:vectex_opt}, the endpoints of a deleted intra-community edge or an inserted cross-community edge may affect their community memberships.
If an affected vertex changes its community, its neighbors outside the target community may also be affected.
Note that any vertex that changes its community has to change its sub-community, since each sub-community is a subset of its community.
Hence, sub-community memberships are also considered in {\tt inc-movement}.

We first introduce the data structures used to maintain a dynamic sub-community.
According to Observation \ref{obs:conn_compoent}, each connected component of a sub-community is treated as a $\gamma$-connected subset.
When edge updates or vertex movements split a sub-community into multiple connected components, we re-assign each resulting component as a new sub-community.

\begin{example}
\label{eg:cc}
  Fig. \ref{fig:removal} shows the sub-community $S_1$ is split into two sub-communities $S_2 = \{v_1,v_3\}$ and $S_3=\{v_2\}$.
  The separation can occur either due to the deletion of edges $(v_1, v_2)$ and $(v_2, v_3)$ during graph updates, as shown in Fig.~\ref{fig:removal}\subref{fig:removal:edge}, or due to the removal of vertex $v_2$ during the movement phase, as shown in Fig.~\ref{fig:removal}\subref{fig:removal:vertex}.
\end{example}

\begin{figure}[t]
  \centering
  \subfloat[A sub-community.]{
    \begin{minipage}[t]{0.155\textwidth}
      \centering
      \includegraphics[width=0.7434\textwidth]{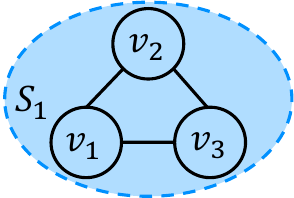}
    \end{minipage}
    \label{fig:removal:graph}
  }
  \subfloat[Delete two edges.]{
    \begin{minipage}[t]{0.155\textwidth}
      \centering
      \includegraphics[width=0.7434\textwidth]{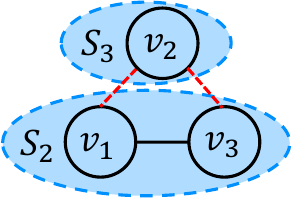}
    \end{minipage}
    \label{fig:removal:edge}
  }
  \subfloat[Move out a vertex.]{
    \begin{minipage}[t]{0.155\textwidth}
      \centering
      \includegraphics[width=0.7434\textwidth]{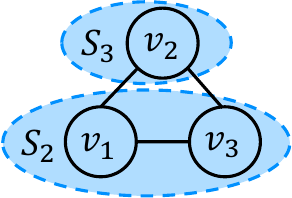}
    \end{minipage}
    \label{fig:removal:vertex}
  }
  \caption{Illustrating the process that a sub-community $S_1$ is split into two sub-communities $S_2$ and $S_3$.}
  \label{fig:removal}
\end{figure}

To preserve the structure under such changes, we leverage the dynamic connected component maintenance technique in \cite{xu2024constant}.
%
%
The graph $G_{\Psi}$ stores the subgraph of $G$ consisting only of intra-sub-community edges based on $s(\cdot)$.

%

\begin{algorithm} [h]
  \small
  \caption{\tt Inc-movement}
  \label{alg:hit:hifm}
  
  \KwIn{$G$, $\Delta G$, $f(\cdot)$, $s(\cdot)$, $\Psi$, $\gamma$}
  
  \KwOut{Updated $f(\cdot)$, $\Psi$, $B$, $K$}
  $A \leftarrow \emptyset$, $B \leftarrow \emptyset$, $K \leftarrow \emptyset$\;
  \For{$(v_i, v_j, \alpha) \in \Delta G$}{
    \If{$\alpha > 0$ and $f(v_i) \neq f(v_j)$}{
      $A.add(v_i)$; $A.add(v_j)$\;
    }
    \If{$\alpha < 0$ and $f(v_i) = f(v_j)$}{  
      $A.add(v_i)$; $A.add(v_j)$\;
    }
    \If{$s(v_i) = s(v_j)$ and $update\_edge\left(G_{\Psi}, (v_i, v_j, \alpha)\right)$}{
      $K.add(v_i)$; $K.add(v_j)$\;
    }
  }
  \While{$A \neq \emptyset$}{
    $v_i \leftarrow A.pop()$\;
    $C^* \leftarrow argmax_{C \in {\mathbb{C}\cup\{\emptyset\}}}{\Delta Q(v_i \to C, \gamma)}$\;
    \If{$\Delta Q(v_i \to C^*, \gamma) > 0$}{
      $f(v_i) \leftarrow C^*$; $B.add(v_i)$; $K.add(v_i)$\;
      \For{$v_j \in N(v_i)$}{
        \If{$f(v_j) \neq C^*$}{
            $A.add(v_j)$\;
        }
      }
      \For{$v_j \in \{ u \in N(v_i) | s(u) = s(v_i)\}$}{
        \If{$update\_edge\left(G_{\Psi}, (v_i, v_j, -w(v_i, v_j))\right)$}{
          $K.add(v_j)$\;
        }
      }
    }
  }
  \KwRet $f(\cdot)$, $\Psi$, $B$, $K$\;
\end{algorithm}

Algorithm \ref{alg:hit:hifm} shows {\tt inc-movement}.
Given an updated graph $G$, a set of graph changes $\Delta G$, community mappings $f(\cdot)$, sub-community mappings $s(\cdot)$, and a CC-index $\Psi$, it first initializes three empty sets: $A$, $B$ and $K$ (line 1).
Here, $A$ keeps the vertices whose community memberships may be changed, $B$ keeps the vertices that have changed their community memberships, and $K$ records the endpoints on edges whose deletion disconnects the connected component in $G_{\Psi}$.
Subsequently, vertices involved in intra-community edge deletion or cross-community edge insertion are added to $A$, and edges in $G_{\Psi}$ are updated according to intra-sub-community changes (lines 2-7) based on Observations \ref{lemma:intra-cross-edge} and \ref{obs:conn_compoent}, respectively.
If an edge update in $G_{\Psi}$ causes a connected component to split (i.e., $update\_edge(\cdot)$ returns $true$), its endpoints are added to $K$ (line 8). 
It then processes vertices in $A$ until the set is empty (line 9).
For each vertex $v_i$, it identifies the target community $C^*$ that yields the highest modularity gain (lines 10-11).
If $\Delta Q(v_i \to C^*) > 0$, $f(v_i)$ is updated to $C^*$, $v_i$ is added to $B$ and $K$, and the neighbors of $v_i$ not in $C^*$ are added to $A$ (lines 12-16), which implements the property in Lemma \ref{lemma:neighbor-change}.
Besides, the intra-sub-community edges involving $v_i$ are deleted from $G_{\Psi}$, and the vertices involved in component splits are added to $K$ (lines 17-19).
Finally, it returns $f(\cdot)$, $\Psi$, $B$,  and $K$ (line 20).

\subsection{Inc-refinement}
\label{subs:ticr}

As discussed in Section \ref{subs:hifm} and Observation \ref{obs:conn_compoent}, we treat each connected component in $G_{\Psi}$ maintained in {\tt inc-movement} as a sub-community.
Therefore, we design {\tt inc-refinement} for re-assigning each new connected component in $G_{\Psi}$ as a sub-community.
Additionally, we attempt to merge singleton sub-communities whose process is the same as the process of the refinement phase in {\tt Leiden} with $G_\Psi$ maintenance.

Algorithm \ref{alg:ticr} presents its pseudocode.
Given an updated graph $G$, community mappings $f(\cdot)$ and sub-community mapping $s(\cdot)$, a CC-index $\Psi$, and a set $K$, it first initializes $R$ as an empty list to track vertices that have changed their sub-communities (line 1).
It then traverses $K$ to identify split connected components in $G_{\Psi}$ using breadth-first search or depth-first search.
All vertices in the split connected component are re-mapped to a new sub-community and added to $R$ (lines 2-3).
For each vertex $v_i \in R$ that is in a singleton sub-community, {\tt inc-refinement} uses a set $\mathcal{T}$ to store the locally optimized neighboring sub-communities of $v_i$ within the same community (lines 4-6). 
If $\mathcal{T} \neq \emptyset$, it attempts to re-assign $v_i$ to a sub-community $S^* \in \mathcal{T}$, which offers the highest modularity gain to eliminate singleton sub-communities (lines 7-9).
Notably, $\Delta Q(v_i \to S, \gamma)$ denotes the modularity gain of moving $v_i$ from $s(v_i)$ to $S$, whose calculation follows the same formula as the standard modularity gain.
If the gain is positive, $s(v_i)$ is updated to $S^*$, and the corresponding intra-sub-community edges are added into $G_{\Psi}$ (lines 10-14).
Finally, {\tt inc-refinement} returns the $s(\cdot)$, $\Psi$, and $R$ (line 15).

\begin{algorithm} [t]
  \small
  \caption{\tt Inc-refinement}
  \label{alg:ticr}
  
  \KwIn{$G$, $f(\cdot)$, $s(\cdot)$, $\Psi$, $K$, $\gamma$}
  \KwOut{Updated $s(\cdot)$, $\Psi$, $R$}

  $R\leftarrow \emptyset$\;
  \For{$v_i \in K$}{
        Map all vertices in the connected component containing $v_i$ into a new sub-community and add them into $R$\;
  }
  \For{$v_i \in R$}{
    \If{$v_i$ is in a singleton sub-community}{
        $\mathcal{T} \leftarrow \{s(v) | v \in N(v_i), f(v)=f(v_i), \Delta Q(s(v) \to \emptyset,\gamma) \leq 0\}$\;
        \If{$\mathcal{T} = \emptyset$}{continue\;}
        $S^* \leftarrow argmax_{S \in \mathcal{T}}{\Delta Q(v_i \to S,\gamma)}$\;
        \If{$\Delta Q(v_i \to S^*,\gamma) > 0$}{
          $s(v_i) \leftarrow S^*$\;
          \For{$v_j \in N(v_i)$}{
            \If{$s(v_i) = s(v_j)$}{
              $update\_edge\left(G_{\Psi}, (v_i, v_j, w(v_i, v_j))\right)$\;
            }
          } 
        }
    }
  }
  \KwRet{$s(\cdot)$, $\Psi$, $R$}\;
\end{algorithm}

\subsection{Inc-aggregation}
\label{subs:da}

Given an updated graph $G$ and its edge changes $\Delta G$, modifications to edges and sub-community memberships are reflected as changes to superedges and supervertices in the supergraph $H$.
Let $s_{pre}(\cdot)$ (resp. $s_{cur}(\cdot)$) denote the vertex-to-supervertex mappings before (resp. after) {\tt inc-refinement}.
Any edge change $(v_i, v_j, \alpha)$ in $\Delta G$ corresponds to a superedge change $(s_{pre}(v_i), s_{pre}(v_j), \alpha)$ in $H$, since the weight of a superedge is the sum of the weights of edges between their sub-communities.
Besides, a vertex $v$ migration from $s_{\text{pre}}(v)$ to $s_{\text{cur}}(v)$ requires updating these weights.
Specifically, the original sub-community $s_{pre}(v)$ must decrease the superedge weights corresponding to the edge incident to $v$, and the new sub-community $s_{cur}(v)$ must increase them under the new assignment.
%

\begin{example}
    \label{eg:da}
    Following Example~\ref{eg:extend}, the initial superedge changes due to edge changes are $(C_1, C_2, 1)$ and $(C_2, C_2, -1)$.
    Then, vertices $v_3$ and $v_4$ move from $C_2$ to $C_1$, and there are three cases:
    \begin{enumerate}
        \item $C_1$ gains edges to the neighbors of $v_3$, resulting in two updates: $(C_1, C_1, 1)$ and $(C_1, C_1, 1)$;
        \item $C_2$ loses edges incident to the neighbors of $v_3$, resulting in $(C_1, C_2, -1)$ and $(C_2, C_2, -1)$;
        \item The effect of $v_4$ is skipped to avoid duplicate updates, since its only neighbor $v_3$ already changed.
    \end{enumerate}
    After compressing the above six superedge changes, we obtain the final superedge changes, which are $(C_1, C_1, 2)$ and $(C_2, C_2, -2)$.
\end{example}

\begin{algorithm} [t]
  \small
  \caption{\tt Inc-aggregation}
  \label{alg:da}
  
  \KwIn{$G$, $\Delta G$, 
  $s_{pre}(\cdot)$, $s_{cur}(\cdot)$, $R$}
  
  \KwOut{$\Delta H$, $s_{pre}(\cdot)$}
  
  $\Delta H \leftarrow \emptyset$\;
  \For{$(v_i, v_j, \alpha) \in \Delta G$}{
    $s_i \leftarrow s_{pre}(v_i)$, $s_j \leftarrow s_{pre}(v_j)$\;
    $\Delta H.add((s_i, s_j, \alpha))$\;
  }

  \For{$v_i \in R$}{
    \For{$v_j \in N(v_i)$}{
      \If{$s_{cur}(v_j) = s_{pre}(v_j)$ or $i < j$}{
          $\Delta H.add((s_{pre}(v_i), s_{pre}(v_j), -w(v_i, v_j)))$\;
          $\Delta H.add((s_{cur}(v_i), s_{cur}(v_j), w(v_i, v_j)))$\;
      }
    }
    $\Delta H.add((s_{pre}(v_i), s_{pre}(v_i), -w(v_i, v_i)))$\;
    $\Delta H.add((s_{cur}(v_i), s_{cur}(v_i), w(v_i, v_i)))$\;
  }
  
  \For{$v_i \in R$}{
    $s_{pre}(v_i) \leftarrow s_{cur}(v_i)$\;
  }
  $Compress(\Delta H)$\;
  \KwRet $\Delta H$, $s_{pre}(\cdot)$\;
\end{algorithm}

Algorithm \ref{alg:da} presents {\tt inc-aggregation}.
Initially, the set of changes $\Delta H$ of $H$ is empty (line 1).
Then, it maps the edge changes $\Delta G$ to superedge changes using $s_{pre}(\cdot)$ (lines 2-4).
Following, it updates superedges for vertices that switch sub-communities by removing edges from the old community and adding edges to the new one. 
For any vertex $v_i$ in $R$, it updates superedges with each neighbor $v_j$ if either $s_{cur}(v_j) = s_{pre}(v_j)$ or $i < j$ to avoid duplicate processing (lines 5-9). 
Besides, it updates the self-loop for the sub-community of $v_i$ (lines 10-11).
Finally, it updates $s_{pre}(\cdot)$ locally to match $s_{cur}(\cdot)$ at the next time step (lines 12-13) and compresses $\Delta H$ by summing the weights of identical superedges (line 14).


\subsection{Overall HIT-Leiden algorithm}


Before presenting our overall HIT-Leiden algorithm, we introduce an optimization technique to further improve the efficiency of the vertices' membership update.
Specifically, when a supervertex changes its community membership, all the lower-level supervertices associated with it have to update their community membership.
As shown in Fig.~\ref{fig:change}, when $v^2_{10}$ changes its community, $v^1_3$ and $v^1_4$ also update their community memberships to the community containing $v^2_{10}$.
However, during the iteration process of {\tt HIT-Leiden}, a supervertex that changes its community does not automatically trigger updates of the community memberships of its constituent lower-level supervertices.

\begin{figure}[t]
  \centering
  \subfloat[Before maintenance.]{
    \begin{minipage}[t]{0.227514285714286\textwidth}
      \centering
      \includegraphics[width=\textwidth]{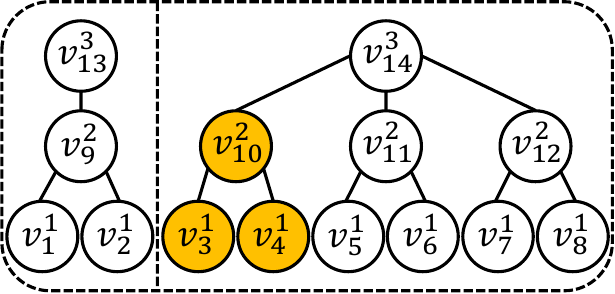}
    \end{minipage}
    \label{fig:change:change_pre}
  }
  \subfloat[After maintenance.]{
    \begin{minipage}[t]{0.241\textwidth}
      \centering
      \includegraphics[width=\textwidth]{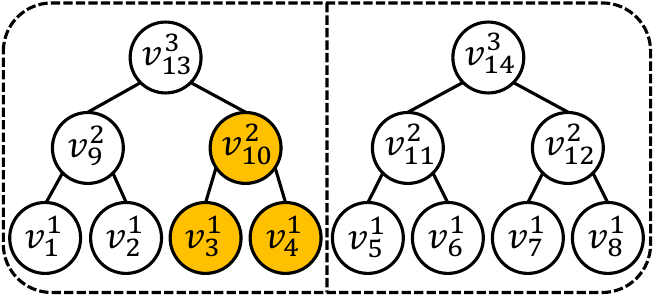}
    \end{minipage}
    \label{fig:change:change_cur}
  }
  \caption{Changes in the hierarchical partitions of Fig. \ref{fig:mod}.}
  \label{fig:change}
\end{figure}


\begin{figure*}[t]
  \centering
  \subfloat[Community maintenance by \texttt{HIT-Leiden}.]{
    \begin{minipage}[t]{0.409854545454545\textwidth}
      \centering
      \includegraphics[width=\textwidth]{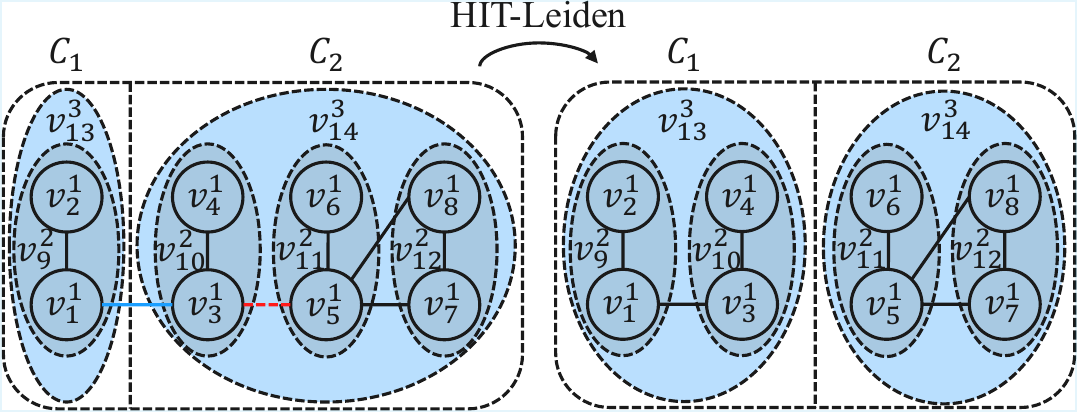}
    \end{minipage}
    \label{fig:ex_hit:ex_hit_graph}
  }
  \begin{tikzpicture}
    \draw[line width=0.65625pt] (0,0) -- (0,2.9);
  \end{tikzpicture}
  \subfloat[The process of \texttt{HIT-Leiden} in iteration two.]{
    \begin{minipage}[t]{0.5352\textwidth}
      \centering
      \includegraphics[width=\textwidth]{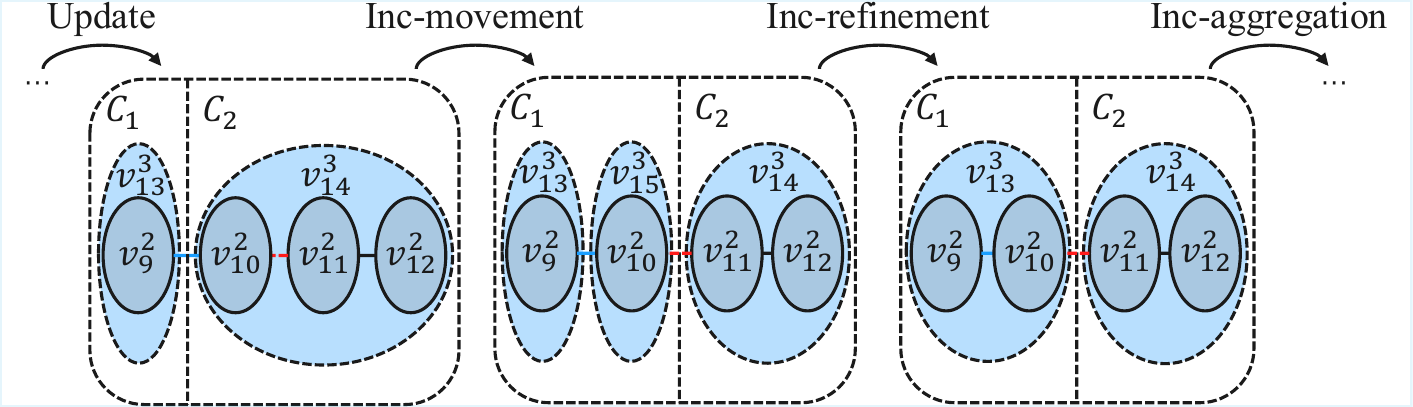}
    \end{minipage}
    \label{fig:ex_hit:ex_hit_l2}
  }
  \caption{An example of \texttt{HIT-Leiden}.}
  \label{fig:ex_hit}
\end{figure*}

\begin{algorithm}[t]
  \small
  \caption{\tt def-update}
  \label{alg:mem_update}
  
  \KwIn{$\{f^P(\cdot)\}$, $\{s^P(\cdot)\}$, $\{B^P\}$, $P$}
  
  \KwOut{Updated $\{f^P(\cdot)\}$}
  
  \For{$p$ from $P$ to $1$}{
    \If{$p \neq P$}{
      \For{$v^p_i \in B^p$}{
        $f^p(v^i_p) \leftarrow f^{p+1}(s^p(v^i_p))$\;
      }
    }
    \If{$p \neq 1$}{
      \For{$v^p_i \in B^p$}{
        $B^{p - 1}.add((s^{p-1})^{-1}(v^p_i))$\;
      }
    } 
  }
  \KwRet $\{f^P(\cdot)\}$\;
\end{algorithm}

To resolve the above inconsistency, we perform a post-processing step to synchronize the community memberships across all levels, as described in Algorithm \ref{alg:mem_update}.
Let $\{B^P\}$ denote a sequence of $P$ sets $\{B^1, \cdots, B^P\}$, $\{s^P(\cdot)\}$ denote a sequence of $P$ adjacent-level supervertex mappings $\{s^1(\cdot), \cdots ,s^P(\cdot)\}$, and $\{f^P(\cdot)\}$ denote a sequence of $P$ community mappings $\{f^1(\cdot), \cdots ,f^P(\cdot)\}$.
Note, each $B^p$ in $\{B^P\}$ collects supervertices at level-$p$ whose community memberships have changed, each $s^p(\cdot)$ in $\{s^P(\cdot)\}$ maps from level-$p$ supervertices to their parent supervertices at level-($p+1$), and each $f^p(\cdot)$ in $\{f^P(\cdot)\}$ maps from level-$p$ supervertices to their communities.
A supervertex is added to $B^p$ for one of two reasons: (1) it changes its community during {\tt inc-movement}, or (2) its higher-level supervertex changes community.
Hence, for each level $p$, {\tt def-update} updates each supervertex in $B^p$ by re-mapping its community membership of its parent using $s^p(\cdot)$ and $f^{p+1}(\cdot)$ when $p \neq P$ (lines 1-4), and adds its constituent vertices to $B^{p-1}$ for the next level updates where $(s^{p-1})^{-1}(\cdot)$ denotes the inverse mapping of $s^{p-1}(\cdot)$ (lines 5-7). 
This algorithm also supports updating the ancestor mappings $\{g^P(\cdot)\}$ from each level supervertex to its level-$P$ ancestor.
%

$\bullet$ \textbf{Overall HIT-Leiden.} 
After introducing all the key components, we present our overall {\tt HIT-Leiden} in 
Algorithm \ref{alg:hit}.
The algorithm proceeds over $P$ hierarchical levels, where each level-$p$ operates on a corresponding supergraph $G^p$.
Besides the community membership $f(\cdot)$, {\tt HIT-Leiden} also maintains supergraphs $\{G^P\}$, community mapping $\{f^P(\cdot)\}$, ancestor mappings $\{g^P(\cdot)\}$, sub-community mappings $\{s^P_{pre}(\cdot)\}$ and $\{s^P_{cur}(\cdot)\}$, and CC-indices $\{\Psi^P\}$ to maintain sub-community memberships for each level.
Note, $\{s^P_{pre}(\cdot)\}$ are the mappings from the previous time step, and $\{s^P_{cur}(\cdot)\}$ are the in-time mappings to track sub-community memberships as they evolve at the current time step.
Specifically, it initializes $\{s^P_{cur}(\cdot)\} = \{s^P_{pre}(\cdot)\}$.
Given the graph change $\Delta G$, it first initializes the first-level update $\Delta G$ to $\Delta G^1$ (line 1).
It then proceeds through $P$ iterations, each including three phases after updating the supergraph $G^p$ (line 3).
\begin{enumerate}
    \item {\tt Inc-movement} (line 4): it re-assigns community memberships of affected vertices to achieve vertex optimality, which yields $f^p(\cdot)$, $\Psi$, $B^p$, and $K$.
    \item {\tt Inc-refinement} (line 5): it re-maps the supervertices of split connected components in $\Psi$ to new sub-communities, producing $s^p_{cur}(\cdot)$, $\Psi$, and $R^p$.
    \item {\tt Inc-aggregation} (line 7): it calculates the next level's superedge changes $\Delta G^{p+1}$, and synchronizes $s^p_{pre}(\cdot)$ to match $s^p_{cur}(\cdot)$.
\end{enumerate}

After $P$ iterations, {\tt def-update} (Algorithm~\ref{alg:mem_update}) synchronizes community mappings $\{f^P(\cdot)\}$ and ancestor mappings $\{g^P(\cdot)\}$ across levels (lines 8-9). 
The final output $f(\cdot)$ is set to $g^1(\cdot)$ (line 10), since $f^1(\cdot)$ could contain disconnected communities with insufficient iterations and $g^1(\cdot) = f^1(\cdot)$ with sufficient iterations.
Finally, we return $\{G^P\}$, $\{f^P(\cdot)\}$, $\{g^P(\cdot)\}$, $\{s^P_{pre}(\cdot)\}$, $\{s^P_{cur}(\cdot)\}$, and $\{\Psi^P\}$ for the next graph evolution (line 11).
%

\begin{algorithm} [t]
  \small
  \caption{{\tt HIT-Leiden}}
  \label{alg:hit}
  \KwIn{$\{G^P\}$, $\Delta G$, $\{f^P(\cdot)\}$, $\{g^P(\cdot)\}$,$\{s^P_{pre}(\cdot)\}$, $\{s^P_{cur}(\cdot)\}$, $\{\Psi^P\}$, $P$, $\gamma$}
  
  \KwOut{$f(\cdot)$, $\{G^P\}$, $\{f^P(\cdot)\}$, $\{g^P(\cdot)\}$, $\{s^P_{pre}(\cdot)\}$, $\{s^P_{cur}(\cdot)\}$, $\{\Psi^P\}$}
  
  $\Delta G^1 \leftarrow \Delta G$\;
  \For{$p=1$ to $P$}{
    $G^p \leftarrow G^p \oplus \Delta G^p$\;

    $f^p(\cdot), \Psi^p, B^p, K \leftarrow \text{\tt inc-movement}(G^p, \Delta G^p, f^p(\cdot), s_{cur}^p(\cdot), \Psi^p, \gamma)$\;
    
    $s^p_{cur}(\cdot), \Psi^p, R^p \leftarrow \text{\tt inc-refinement}(G^p, f^p(\cdot), s^p_{cur}(\cdot), \Psi^p, K, \gamma)$\;
    \If{$p < P$}{
        $\Delta G^{p+1}, s^p_{pre}(\cdot) \leftarrow \text{\tt inc-aggregation}(G^p, \Delta G^p, s^p_{pre}(\cdot), s^p_{cur}(\cdot), R^p)$\;
    }
  }
  $\{f^P(\cdot)\} \leftarrow \text{\tt def-update}(\{f^P(\cdot)\}, \{s^P_{cur}(\cdot)\}, \{B^P\}, P)$\;
  $\{g^P(\cdot)\} \leftarrow \text{\tt def-update}(\{g^P(\cdot)\}, \{s^P_{cur}(\cdot)\}, \{R^P\}, P)$\;
  $f(\cdot) \leftarrow g^1(\cdot)$\;
  \KwRet $f(\cdot)$, $\{G^P\}$, $\{f^P(\cdot)\}$, $\{g^P(\cdot)\}$, $\{s^P_{pre}(\cdot)\}$, $\{s^P_{cur}(\cdot)\}$, $\{\Psi^P\}$\;
\end{algorithm}

\begin{example}
  \label{eg:l2process}
  Consider the result in Fig. \ref{fig:hie_structural}. 
  The graph undergoes an edge deletion $(v^1_3, v^1_5, -1)$ and an edge insertion $(v^1_1, v^1_3, 1)$.
  Resulting community and sub-community changes are shown in Fig.~\ref{fig:ex_hit}, with hierarchical changes in Fig.~\ref{fig:change}.
  Take the second iteration as an example. 
  In {\tt inc-movement}, the supervertex $v^2_{10}$ is reassigned to $v^3_{15}$ due to disconnection, and migrates from community $C_2$ to $C_1$.
  In {\tt inc-refinement}, $v^2_{10}$ is merged into $v^3_{13}$.
  Then, {\tt inc-aggregation} calculates superedge changes for level-$3$, including edge insertion $(v^3_{13}, v^3_{13}, 2)$ and edge deletions $(v^3_{14}, v^3_{14}, -2)$.
\end{example}


$\bullet$ \textbf{Complexity analysis.}
We now analyze the time and space complexity of {\tt HIT-Leiden} over $P$ iterations.
Let $\Gamma^p$ denote the set of supervertices involved in superedge changes, let $\Lambda^p$ track the supervertices that change their communities or sub-communities at level-$p$, and let $h^p$ denote the average vertex depth in the spanning tree in \cite{xu2024constant} at level-$p$.
Typically, $h^p$ is treated as a small factor.
For each level-$p$, {\tt inc-movement} and {\tt inc-refinement} complete in $O\left(|N_2(\Gamma^p)|+h^p\cdot|\Gamma^p| + |N_2(\Lambda^p)|+h^p\cdot|\Lambda^p|\right)$, and {\tt inc-aggregation} completes in $O(|N(\Gamma^p)| + |N(\Lambda^p)|)$.
Besides, the time cost of {\tt def-update} is $O\left(\sum_{p=1}^P |\Lambda^p|\right)$.
Hence, the total time cost of {\tt HIT-Leiden} can be expressed as $O(|N_2(\texttt{CHANGED})| + |N_2(\texttt{AFF})|)$, as analyzed in Section \ref{subsec:vectex_opt}.
As a result, our {\tt HIT-Leiden} is relatively bounded to Leiden.
In the worst case, \texttt{HIT-Leiden} degenerates to Leiden with a full rebuild of the CC-index at each level.
The total time is upper-bounded by $O\left(\sum_{p=1}^{P}\left(|V^p| + (1+h^p) \cdot |E^p|\right)\right)$.
For space complexity, at each level-$p$, {\tt HIT-Leiden} stores the aggregated graph, community/sub-community assignments, auxiliary update sets, and the CC-index.
Since the CC-index requires linear space with respect to the number of vertices and edges, the space cost at level-$p$ is $O(|V^p|+|E^p|)$.
%
Thus, the total space cost is at most
$O\left(\sum_{p=1}^{P}\left(|V^p| + |E^p|\right)\right)$. 

\section{Experiments}
\label{sec:exp}

We now present our experimental results.
Section \ref{subc:setup} introduces the experimental setup. 
Sections \ref{subs:exp:modularity} and \ref{subs:exp:time} evaluate the effectiveness and efficiency of {\tt HIT-Leiden}, respectively.

\input{sections/experiments/6.2-figure.tex}
\input{sections/experiments/6.3-figure.tex}
\subsection{Setup}
\label{subc:setup}

{\bf Datasets.} 
We use four real-world dynamic datasets, including 
\emph{dblp-coauthor}\footnotemark[2] (academic collaboration), 
\emph{yahoo-song}\footnotemark[2] (user-song interactions), 
\emph{sx-stackoverflow}\footnotemark[3] (developer Q\&A), and 
\emph{risk} (financial transactions) provided by \Acom.
All these dynamic edges are associated with real timestamps.
We also use one static dataset \emph{it-2004}\footnotemark[4] (a large-scale web graph), but randomly insert or delete some edges to simulate a dynamic graph.
All the graphs are treated as undirected graphs.
For each real-world dynamic graph, we collect a sequence of batch updates by sorting the edges in ascending order of their timestamps; for \emph{it-2004}, which lacks timestamps, we randomly shuffle its edge order.
Table~\ref{tab:datasets} summarizes the key statistics of the above datasets.

\begin{table}[t]
  \centering
  \setlength{\belowcaptionskip}{0pt}
  \setlength{\abovecaptionskip}{0pt}
  \caption{Datasets used in our experiments.}
  \label{tab:datasets}
  \begin{tabular}{ c | c | c | c | c }
  \hline
  \textbf{Dataset} & \textbf{Abbr.} & $|V|$ & $|E|$ & \textbf{Timestamp}  \\ 
    \hline
    \hline
    dblp-coauthor & {\tt DC} & 1.8M & 29.4M & Yes \\
    \hline
    yahoo-song & {\tt YS} & 1.6M & 256.8M & Yes \\
    \hline
    sx-stackoverflow & {\tt SS} & 2.6M & 63.4M & Yes \\
    \hline
    it-2004 & {\tt IT} & 41.2M & 1.0B & No \\
    \hline
    risk & {\tt RS} & 201.0M & 4.0B & Yes \\
    \hline
  \end{tabular}
  
\end{table}

\footnotetext[2]{http://konect.cc/networks/}
\footnotetext[3]{https://snap.stanford.edu/data/}
\footnotetext[4]{https://networkrepository.com/}

{\bf Algorithms.}
We test the following maintenance algorithms:
\begin{itemize}
  \item {\tt ST-Leiden}: A naive baseline that executes the static Leiden algorithm from scratch when the graph changes.
  
  \item {\tt ND-Leiden}: A simple maintenance algorithm in~\cite{sahu2024starting}, which processes all vertices during the movement phase, initialized with previous community memberships.
  
  \item {\tt DS-Leiden}: A maintenance algorithm based on~\cite{sahu2024starting}, which uses the delta-screening technique~\cite{zarayeneh2021delta} to restrict the number of vertices considered in the movement phase.
  
  \item {\tt DF-Leiden}: An advanced maintenance algorithm from~\cite{sahu2024starting}, which adopts the dynamic frontier approach~\cite{sahu2024df} to support localized updates.
  
  \item {\tt HIT-Leiden}: Our proposed method.
\end{itemize}

{\bf Dynamic graph settings.}
As the temporal span varies across datasets (e.g., 62 years for \emph{dblp-coauthor} versus 8 years for \emph{sx-stackoverflow}), we apply a sliding edge window, avoiding reliance on fixed valid time intervals that are hard to standardize.
Initially, we construct a static graph using the first 80\% of edges.
Then, we select a window size $b \in \{10, 10^2, 10^3, 10^4, 10^5\}$, denoting the number of updated edges in an updated batch.
Next, we slide this window $r=9$ times, so we update 9 batches of edges for each dataset.
In the figures, batch $0$ denotes the initial state before the first evaluated update.
Note that by default, we set $b=10^3$.

All the algorithms are implemented in C++ and compiled with the gcc 8.3.0 compiler using the same optimization level.
We set $\gamma = 1$ and use $P = 10$ iterations.
Before running the Leiden community maintenance algorithms, we obtain the communities by running the Leiden algorithm, and {\tt HIT-Leiden} requires an additional procedure to build auxiliary structures.
Due to the limited number of iterations, the community structure has not fully converged, so the maintenance algorithms usually take more time in the first two batches than in other batches.
Therefore, we exclude the first two batches from efficiency evaluations.
All experiments in this paper are conducted only on a Linux server running Debian 5.4.56, equipped with an Intel(R) Xeon(R) Platinum 8336C CPU @ 2.30GHz and 2.0 TB of RAM.

\subsection{Effectiveness evaluation}
\label{subs:exp:modularity}

To evaluate the effectiveness of different maintenance algorithms, we compare the modularity value and proportion of subpartition $\gamma$-dense communities for their returned communities.
%

$\bullet$ {\bf Modularity.}
Fig. \ref{fig:exp_shift_modularity} depicts the average modularity values of all the maintenance algorithms, where the batch size ranges from 10 to $10^5$.
Fig. \ref{fig:exp_modularity} depicts the modularity value across all 9 update batches, where the batch size is fixed as 1,000.
Across all datasets, the expected fluctuation in modularity for {\tt ST-Leiden} is around $0.02$ due to its inherent randomness.
These maintenance algorithms achieve equivalent quality in modularity, since the difference in their modularity values is within $0.01$.
Overall, our {\tt HIT-Leiden} achieves comparable modularity with other methods.


$\bullet$ {\bf Proportion of subpartition $\gamma$-density.} 
After running {\tt HIT-Leiden}, for each returned community, we try to re-find its $\gamma$-order such that any intermediate vertex set in the $\gamma$-order is locally optimized, according to Definition \ref{def:sgd}.
If we can find a valid $\gamma$-order for a community, we classify it as a subpartition $\gamma$-dense community.
We report the proportion of subpartition $\gamma$-dense communities in Fig. \ref{fig:gamma-desnity}.
The proportions of subpartition $\gamma$-dense communities are consistently high across all Leiden-based algorithms. 
For each dataset, the differences among these algorithms are very small, with fluctuations around 0.0001.
Thus, {\tt HIT-Leiden} achieves a comparable percentage of subpartition $\gamma$-density with others.

\subsection{Efficiency evaluation}
\label{subs:exp:time}

In this section, we first present the overall efficiency results, then analyze the time cost of each component, and finally evaluate the effects of some hyperparameters.


$\bullet$ {\bf Overall results.} 
Fig. \ref{fig:exp_overal_time} presents the overall efficiency results where $b$ is set to its default value $1,000$.
Clearly, {\tt HIT-Leiden} achieves the best efficiency on datasets, especially on the {\tt it-2004} dataset, since it is up to three orders of magnitude faster than the state-of-the-art algorithms.
That is mainly because the ratio of updated edges to total edges in {\tt it-2004} is smaller than those in {\tt dblp-coauthor}, {\tt yahoo-song}, and {\tt sx-stackoverflow}. 

$\bullet$ \textbf{Time cost of different components in \texttt{ HIT-Leiden}.} 
There are three key components, i.e., {\tt inc-movement}, {\tt inc-refinement}, and {\tt inc-aggregation}, in {\tt HIT-Leiden}.
We evaluate the proportion of time cost for each component and present the results in Fig. \ref{fig:exp_proportion}.
Note that some operations (e.g., {\tt def-update} in {\tt HIT-Leiden}) may not be included by the above three components, so we put them into the "Others" component.
Notably, in {\tt HIT-Leiden}, the refinement phase contributes minimally to the overall runtime. 
Besides, the combined proportion of time spent in its movement and aggregation phases is comparable to that of other algorithms.
{\tt Inc-movement}, {\tt inc-refinement}, and {\tt inc-aggregation} in {\tt HIT-Leiden} consistently outperform their counterparts in other algorithms across all datasets, achieving lower runtime costs according to Fig. \ref{fig:exp_overal_time}.

$\bullet$ {\bf Effect of $b$.}
We vary the batch size $b\in \{10, 10^2, 10^3, 10^4$, $10^5\}$ and report the efficiency in Fig. \ref{fig:exp_shift_time}.
We see that {\tt HIT-Leiden} is up to five orders of magnitude faster than other algorithms.
Also, its runtime decreases substantially as $b$ becomes smaller since it is a relatively bounded algorithm.
In contrast, {\tt ND-Leiden}, {\tt DS-Leiden}, and {\tt DF-Leiden} still need to process the entire graph when processing a new batch.

$\bullet$ {\bf Effect of $r$.}
Recall that after fixing the batch size $b$, we update the graph for $r$ batches.
Fig. \ref{fig:exp_time} shows the efficiency, where $b$ is fixed as 1,000, but $r$ ranges from 1 to 9.
We observe that the incremental speedup is limited in the first few batches because $P=10$ is small, and additional iterations may slightly improve the community membership.
As a result, all the maintenance algorithms often require more time for the second batch to adjust the community structure.
Once high-quality community structure is established, the speedup becomes significant.
In addition, {\tt HIT-Leiden} incurs a slightly higher runtime at these first batches due to recording more information and constructing the CC-index.


\section{Conclusions}
\label{sec:con}

In this paper, we develop an efficient algorithm for maintaining Leiden communities in a dynamic graph.
We first theoretically analyze the boundedness of existing algorithms and how supervertex behaviors affect community membership under graph update.
Building on these analyses, we further develop a relatively bounded algorithm, called {\tt HIT-Leiden}, which consists of three key components, i.e., {\tt inc-movement}, {\tt inc-refinement}, and {\tt inc-aggregation}.
Extensive experiments on five real-world dynamic graphs show that {\tt HIT-Leiden} not only preserves the properties of Leiden and achieves comparable modularity quality with Leiden, but also runs faster than state-of-the-art competitors. 
In future work, we will extend our algorithm to handle directed graphs and also evaluate it in a parallel environment.
%


\bibliographystyle{IEEEtran}
\bibliography{sample}

\appendices
\label{sec:appendix}

\section{Proof of lemmas}
\label{sec:append:proof}

\subsection{Proof of Lemma \ref{lemma:neighbor-change}}
\begin{proof}
Assuming $v$ changes its community from $C_i$ to $C$, there are three cases:
\begin{enumerate}
    \item For each neighbor $v_{i}$ in $C_i$, the edge $(v, v_i)$ is a {\it deleted intra-community} edge and an inserted cross-community edge; 
    \item For each neighbor $v_j$ in $C$, the edge $(v, v_j)$ is a deleted cross-community edge and an inserted intra-community edge; 
    \item For each other neighbor $v_k$, edge $(v, v_k)$ is a deleted cross-community edge and an {\it inserted cross-community} edge.
\end{enumerate}
Since only the first and third cases meet the conditions in Observation \ref{lemma:intra-cross-edge}, all the neighbors of $v$ that are not in $C$ are likely to change their communities. 
\end{proof}

\subsection{Proof of Lemma \ref{lem:intra_sub_e_del}}
\begin{proof}
\label{proof:intra_sub_e_del}
  We analyze the modularity gain $\Delta M(v \to \emptyset, \gamma)$ for any vertex $v$, which denotes the modularity gain of moving $v$ from the intermediate subsequence $I$ to $\emptyset$, whose calculation follows the same formula as the standard modularity gain.
  
  According to Definition \ref{def:gcc}, if $\Delta M(v \to \emptyset, \gamma) > 0$, the intermediate subsequence $I$ could not be $\gamma$-connected and $v$ has to leave $I$.
  It is different from maintaining vertex optimality (mentioned in Definition \ref{def:vertex-optimality}): If there exists a community $C'$ such that the modularity gain of moving $v$ from its community $C$ to $C'$ is positive, $v$ is not locally optimized and has to be removed from $C$. 

  \underbar{\textbf{Case 1: $v_i$ is inserted into $S$ after $v_j$,}}
  i.e., $v_j \in I_i$. 
  The old modularity gain $M_{old}(v_i \to \emptyset, \gamma) < 0$ before deletion is: 
  \begin{equation}
    \begin{aligned}
      M_{old}(v_i \to \emptyset, \gamma)&=-\frac{w(v_i, U_i)}{m} + \frac{\gamma \cdot d(v_i) \cdot d(U_i)}{2m^2} \leq 0.
    \end{aligned}
    \label{equ:old_mg_intra_sub_vi}
  \end{equation}
  
  Where $U_i = I_i \setminus \{v_i\}$.
  We multiply right side of \eqref{equ:old_mg_intra_sub_vi} by $2m^2$ and obtain $X_{(\text{\ref{equ:old_mg_intra_sub_vi}})}$:
  \begin{equation}
    X_{(\text{\ref{equ:old_mg_intra_sub_vi}})} = -2m \cdot w(v_i, U_i) + \gamma \cdot d(v_i) \cdot d(U_i) \leq 0
    \label{equ:old_mg_intra_sub_vi_numerator}
  \end{equation}
  
  After the deletion, the new modularity gain $M_{new}(v_i \to \emptyset, \gamma)$ formulates:
  \begin{equation}
    \begin{aligned}
      \Delta M_{new}(v_i \to \emptyset, \gamma) &= - \frac{w(v_i, U_i) - 2 \alpha}{m - \alpha} \\
      &+ \frac{\gamma \cdot (d(v_i) - \alpha) \cdot (d(U_i) - \alpha)}{2(m - \alpha)^2}.
    \end{aligned}
    \label{equ:mg_intra_sub_vi}
  \end{equation}
  
  We multiply right side of \eqref{equ:mg_intra_sub_vi} by $2(m-\alpha)^2$ and obtain $Y_{(\text{\ref{equ:mg_intra_sub_vi}})}$:
  
  \begin{equation}
    \begin{aligned}
      Y_{(\text{\ref{equ:mg_intra_sub_vi}})} &= - 2(m - \alpha) \cdot (w(v_i, U_i) - 2 \alpha) \\
      &+ \gamma \cdot (d(v_i) - \alpha) \cdot (d(U_i) - \alpha) \\
      &= X_{(\text{\ref{equ:old_mg_intra_sub_vi}})} + \alpha\cdot(4m + 2w(v_i, U_i) - 4a - \gamma\cdot(d(I_i) - \alpha)) \\
      &< X_{(\text{\ref{equ:old_mg_intra_sub_vi}})} + \alpha\cdot(4m + 2w(v_i, U_i)) \\
    \end{aligned}
    \label{equ:new_mg_intra_sub_vi_numerator}
  \end{equation}

  If $X_{(\text{\ref{equ:old_mg_intra_sub_vi}})} + \alpha\cdot(4m + 2w(v_i, U_i)) > 0$, $\Delta M_{new}(v_i \to \emptyset, \gamma)$ could be positive; Otherwise, $\Delta M_{new}(v_i \to \emptyset, \gamma)$ must be non-positive.
  Therefore, $v_i$ could be removed from its sub-community only if $\alpha > \frac{2m \cdot w(v_i, U_i) - \gamma \cdot d(v_i) \cdot d(U_i)}{4m + 2w(v_i, U_i)}$.
  
  \underbar{\textbf{Case 2: $v_j$ is inserted into $S$ before $v_i$.}}
  In this case, we have $v_j \in I_j$, $v_i \notin I_j$, and the edge deletion does not affect intra-edges within $U_j$. The old modularity gain $M_{old}(v_i \to \emptyset, \gamma) < 0$ before deletion is: 
  \begin{equation}
    \begin{aligned}
      M_{old}(v_j \to \emptyset, \gamma)&=-\frac{w(v_j, U_j)}{m} + \frac{\gamma \cdot d(v_j) \cdot d(U_j)}{2m^2}.
    \end{aligned}
    \label{equ:old_mg_intra_sub_vj}
  \end{equation}
  
  We multiply right side of \eqref{equ:old_mg_intra_sub_vj} by $2m^2$ and obtain $X_{(\text{\ref{equ:old_mg_intra_sub_vj}})}$:
  \begin{equation}
    X_{(\text{\ref{equ:old_mg_intra_sub_vj}})} = -2m \cdot w(v_j, U_j) + \gamma \cdot d(v_j) \cdot d(U_j) < 0
    \label{equ:old_mg_intra_sub_vj_numerator}
  \end{equation}
  
  The new modularity gain after the edge deletion becomes:

  \begin{equation}
    \begin{aligned}
      \Delta M_{new}(v_j \to \emptyset, \gamma) &= -\frac{w(v_j, U_j)}{m - \alpha} \\
      &+ \frac{\gamma \cdot (d(v_j) - \alpha) \cdot d(U_j)}{2(m - \alpha)^2}\\
    \end{aligned}
    \label{equ:mg_intra_sub_vi_former}
  \end{equation}
    
  We multiply right side of \eqref{equ:mg_intra_sub_vi_former} by $2(m-\alpha)^2$ and obtain $Y_{(\text{\ref{equ:mg_intra_sub_vi_former}})}$:
  
  \begin{equation}
    \begin{aligned}
      Y_{(\text{\ref{equ:mg_intra_sub_vi_former}})} &= - 2(m - \alpha) \cdot w(v_j, U_j) + \gamma \cdot (d(v_j) - \alpha) \cdot d(U_j) \\
      & = X_{(\text{\ref{equ:old_mg_intra_sub_vj}})} +2\alpha \cdot \big(w(v_j, U_j) - \gamma \cdot d(U_j)\big) \\
      & < X_{(\text{\ref{equ:old_mg_intra_sub_vj}})} +2\alpha \cdot w(v_j, U_j)
    \end{aligned}
    \label{equ:new_mg_intra_sub_vi_numerator}
  \end{equation}
  
  Hence, $v_j$ could be removed from its sub-community only if $\alpha >  m - \frac{\gamma \cdot d(v_j) \cdot d(U_j)}{2w(v_j,U_j)}$.
  
  \underbar{\textbf{Generalization to other vertices.}}
  Consider other vertices $v_k$ and $v_l$ such that $v_k \in S_i$, $k \neq i, j$ and $v_l \notin S_i$.
  The old modularity gains $M_{old}(v_k \to \emptyset, \gamma) < 0$ and $M_{old}(v_l \to \emptyset, \gamma) < 0$ before deletion are: 
  \begin{equation}
    \begin{aligned}
      M_{old}(v_k \to \emptyset, \gamma) &= -\frac{w(v_k, U_k)}{m} + \frac{\gamma \cdot d(v_k) \cdot d(U_k)}{2m^2}.
    \end{aligned}
    \label{equ:old_mg_intra_sub_vk}
  \end{equation}

  \begin{equation}
    \begin{aligned}
      M_{old}(v_l \to \emptyset, \gamma) &= -\frac{w(v_l, U_l)}{m} + \frac{\gamma \cdot d(v_l) \cdot d(U_l)}{2m^2}.
    \end{aligned}
    \label{equ:old_mg_intra_sub_vl}
  \end{equation}
  
  We multiply right side of \eqref{equ:old_mg_intra_sub_vk} and (\ref{equ:old_mg_intra_sub_vl}) by $2m^2$ respectively to obtain $X_{(\text{\ref{equ:old_mg_intra_sub_vk}})}$ and $X_{(\text{\ref{equ:old_mg_intra_sub_vl}})}$:
  
  \begin{equation}
    X_{(\text{\ref{equ:old_mg_intra_sub_vk}})} = -2m \cdot w(v_k, U_k) + \gamma \cdot d(v_k) \cdot d(U_k) \leq 0
    \label{equ:old_mg_intra_sub_vk_numerator}
  \end{equation}

  \begin{equation}
    X_{(\text{\ref{equ:old_mg_intra_sub_vl}})} = -2m \cdot w(v_l, U_l) + \gamma \cdot d(v_l) \cdot d(U_l) \leq 0
    \label{equ:old_mg_intra_sub_vl_numerator}
  \end{equation}

  After the edge deletion, their new modularity gains are satisfied:

  \begin{equation}
    \begin{aligned}
      \Delta M_{new}(v_{k} \to \emptyset, \gamma) &\leq - \frac{w(v_k, U_k)}{m - \alpha} + \frac{\gamma \cdot d(v_k) \cdot d(U_k)}{2(m - \alpha)^2}. \\
    \end{aligned}
    \label{equ:new_mg_intra_sub_vk}
  \end{equation}
  

  \begin{equation}
    \begin{aligned}
      \Delta M_{new}(v_{l} \to \emptyset, \gamma) &= - \frac{w(v_l, U_l)}{m - \alpha} + \frac{\gamma \cdot d(v_l) \cdot d(U_l) }{2(m - \alpha)^2}. \\
    \end{aligned}
    \label{equ:new_mg_intra_sub_vl}
  \end{equation}


  $v_k$ could be merged before $v_i$ and $v_j$, between $v_i$ and $v_j$, as well as after $v_i$ and $v_j$. 
  $\Delta M_{new}(v_{k} \to \emptyset, \gamma)$ can be formulated as follows:
  
  (1) $v_k$ is merged before $v_i$ and $v_j$:
      \begin{equation}
          \Delta M_{new}(v_{k} \to \emptyset, \gamma) = - \frac{w(v_k, U_k)}{m - \alpha} + \frac{\gamma \cdot d(v_k) \cdot d(U_k)}{2(m - \alpha)^2};
        \label{equ:new_mg_intra_sub_vk_case_1}
      \end{equation}

  (2)$v_k$ is merged between $v_i$ and $v_j$:
      \begin{equation}
          \Delta M_{new}(v_{k} \to \emptyset, \gamma) = - \frac{w(v_k, U_k)}{m - \alpha} + \frac{\gamma \cdot d(v_k) \cdot (d(U_k) - \alpha)}{2(m - \alpha)^2};
        \label{equ:new_mg_intra_sub_vk_case_1}
      \end{equation}

  (3) $v_k$ is merged after $v_i$ and $v_j$:
      \begin{equation}
          \Delta M_{new}(v_{k} \to \emptyset, \gamma) = - \frac{w(v_k, U_k)}{m - \alpha} + \frac{\gamma \cdot d(v_k) \cdot (d(U_k) - 2\alpha)}{2(m - \alpha)^2}.
        \label{equ:new_mg_intra_sub_vk_case_1}
      \end{equation}
  
  Therefore, the equivalent of \eqref{equ:new_mg_intra_sub_vk} holds if and only if $v_k$ is merged before $v_i$ and $v_j$. 
  Then, We multiply right side of \eqref{equ:new_mg_intra_sub_vk} and (\ref{equ:new_mg_intra_sub_vl}) by $2(m-\alpha)^2$ respectively and obtain $Y_{(\text{\ref{equ:new_mg_intra_sub_vk}})}$ and $Y_{(\text{\ref{equ:new_mg_intra_sub_vl}})}$:

  \begin{equation}
    \begin{aligned}
      Y_{(\text{\ref{equ:new_mg_intra_sub_vk}})} &= - 2(m - \alpha) \cdot w(v_k, U_k) \\
      &+ \gamma \cdot d(v_k) \cdot d(U_k) \\
      &= X_{\text{\ref{equ:old_mg_intra_sub_vk_numerator}}} + 2\alpha \cdot w(v_k, U_k),
    \end{aligned}
    \label{equ:new_mg_intra_sub_vk_numerator}
  \end{equation}

  \begin{equation}
    \begin{aligned}
      Y_{(\text{\ref{equ:new_mg_intra_sub_vl}})} &= X_{\text{\ref{equ:old_mg_intra_sub_vl_numerator}}} + 2\alpha \cdot w(v_l, U_l),
    \end{aligned}
    \label{equ:new_mg_intra_sub_vl_numerator}
  \end{equation}

  Only if $\alpha > m - \frac{\gamma \cdot d(v_k) \cdot d(U_k) }{2w(v_k, U_k)}$, $v_k$ could be removed from its sub-community; $v_l$ should be removed from its sub-community if and only if $\alpha > m - \frac{\gamma \cdot d(v_l) \cdot d(U_l)}{2w(v_l, U_l)}$.

\end{proof}

\subsection{Proof of Lemma \ref{lem:cross_sub_e_del}}
\begin{proof}
  \label{proof:cross_sub_e_del}
  
  We adopt the same notations as in the proof of Lemma \ref{lem:intra_sub_e_del}, with the exception that $v_k$ now denotes a vertex residing in the same sub-community as either $v_i$ or $v_j$.
  Based on this setup, the modularity gain after the edge deletion is shown as follows.

  \underbar{\textbf{Case 1: Consider the endpoint $v_i$:}}
  \begin{equation}
    \begin{aligned}
      \Delta M_{new}(v_{i} \to \emptyset, \gamma) &= -\frac{w(v_i, U_i)}{m - \alpha} \\
      & + \frac{\gamma \cdot (d(v_i) - \alpha) \cdot d(U_i)}{2(m - \alpha)^2}.
    \end{aligned}
    \label{equ:mg_inter_sub_vi_former}
  \end{equation}

  We multiply right side of \eqref{equ:mg_inter_sub_vi_former} by $2(m-\alpha)^2$ and obtain $Y_{(\text{\ref{equ:mg_inter_sub_vi_former}})}$:
  \begin{equation}
    \begin{aligned}
      Y_{(\text{\ref{equ:mg_inter_sub_vi_former}})} &= -2(m-\alpha) \cdot w(v_i, U_i)\\
      &+ \gamma \cdot (d(v_i) - \alpha) \cdot d(U_i) \\
      &= X_{(\text{\ref{equ:old_mg_intra_sub_vi}})} + \alpha \cdot (2w(v_i, U_i)- \gamma \cdot d(U_i)) \\ 
      &< X_{(\text{\ref{equ:old_mg_intra_sub_vi}})} + \alpha \cdot 2w(v_i, U_i)
    \end{aligned}
    \label{equ:new_mg_inter_sub_vi_numerator}
  \end{equation}

  Only if $\alpha > m - \frac{\gamma \cdot d(v_i) \cdot d(U_i)}{2w(v_i, U_i)}$, $v_i$ could be removed from its sub-community. 
  $v_j$ holds similar behavior.

  \underbar{\textbf{Case 2: Consider the vertex $v_k \in S_i \cup S_j$, $k \neq i, j$:}}

  \begin{equation}
    \begin{aligned}
      \Delta M_{new}(v_k \to \emptyset, \gamma) &\leq -\frac{w(v_k, U_k)}{m - \alpha} + \frac{\gamma \cdot d(v_k) \cdot d(U_k)}{2(m - \alpha)^2}.
    \end{aligned}
    \label{equ:mg_inter_sub_vk_former}
  \end{equation}

  For \eqref{equ:mg_inter_sub_vk_former}, $v_k$ could be merged before $v_i$ or $v_j$, as well as after $v_i$ or $v_j$. 
  Its equivalent holds if and only if $v_k$ is merged before $v_i$ or $v_j$. 
  We multiply right side of \eqref{equ:mg_inter_sub_vk_former} by $2(m-\alpha)^2$ and obtain $Y_{(\text{\ref{equ:mg_inter_sub_vk_former}})}$:

  \begin{equation}
    \begin{aligned}
      Y_{(\text{\ref{equ:mg_inter_sub_vk_former}})} &= -2(m-\alpha) \cdot w(v_k, U_k) + \gamma \cdot d(v_k) \cdot d(U_k) \\
      &= X_{(\text{\ref{equ:old_mg_intra_sub_vk}})} + 2 \alpha\cdot w(v_k, U_k)
    \end{aligned}
    \label{equ:new_mg_inter_sub_vk_numerator}
  \end{equation}
  
  Only if $\alpha > m - \frac{\gamma \cdot d(v_k) \cdot d(U_k)}{2w(v_k, U_k)}$, $v_k$ could be removed from its sub-community.
  

  \underbar{\textbf{Case 3: Consider the vertex $v_l \notin S_i \cup S_j$:}}
  
  

  \begin{equation}
    \begin{aligned}
      \Delta M_{new}(v_l \to \emptyset, \gamma) &= -\frac{w(v_l, U_l)}{m - \alpha} + \frac{\gamma \cdot d(v_l) \cdot d(U_l)}{2(m - \alpha)^2}.
    \end{aligned}
  \end{equation}
  
  Similar to \textbf{Case 2}, if and only if $\alpha > m - \frac{\gamma \cdot d(v_l) \cdot d(U_l)}{2w(v_l, U_l)}$, $v_l$ should be removed from its sub-community.

\end{proof}

\subsection{Proof of Lemma \ref{lem:sub_e_ins}}
\begin{proof}
  \label{proof:sub_e_ins}
{
  First, we analyze the \textbf{insertion of intra-sub-community edges}. 
  We adopt the same notations as in the proof of Lemma \ref{lem:intra_sub_e_del}.
  Based on this setup, the modularity gain after the edge insertion is shown as follows.

  \underbar{\textbf{Case 1: Consider the endpoint $v_i$}}, which is the latter merged endpoint:
  \begin{equation}
    \begin{aligned}
      \Delta M_{new}(v_{i} \to \emptyset, \gamma) &= -\frac{w(v_i, U_i) + 2\alpha}{m + \alpha} \\
      &+ \frac{\gamma \cdot (d(v_i)+ \alpha) \cdot (d(U_i)+ \alpha)}{2( m + \alpha)^2}.
    \end{aligned}
    \label{equ:mg_intra_add_vj_latter}
  \end{equation}

  We multiply right side of \eqref{equ:mg_intra_add_vj_latter} by $2(m+\alpha)^2$ and obtain $Y_{(\text{\ref{equ:mg_intra_add_vj_latter}})}$:
  \begin{equation}
    \begin{aligned}
      Y_{(\text{\ref{equ:mg_intra_add_vj_latter}})} &= -2(m+\alpha) (w(v_i, U_i) + 2\alpha)\\
      &+ \gamma \cdot \left(d(v_i)+ \alpha) \cdot (d(U_i)+ \alpha\right) \\
      &= X_{(\text{\ref{equ:old_mg_intra_sub_vi}})} +  \alpha\cdot\big(\gamma \cdot(d(I_i) + \alpha) - 2 w(v_i, U_i) - 4\alpha - 4m\big)\\
      &< X_{(\text{\ref{equ:old_mg_intra_sub_vi}})} +  \alpha\cdot\big(\gamma \cdot(d(I_i) + \alpha) - 4m\big)
    \end{aligned}
    \label{equ:mg_intra_add_vj_latter_numerator}
  \end{equation}

  

  Obviously, only if $\gamma \cdot(d(I_i) + \alpha) - 4m > 0$, i.e., $\alpha > \frac{4}{\gamma}m - d(I_i)$, $Y_{(\text{\ref{equ:mg_intra_add_vj_latter}})}$ could be positive.
  %

  \underbar{\textbf{Case 2: Consider the endpoint $v_j$}}, which is the former merged endpoint:

  \begin{equation}
    \begin{aligned}
      \Delta M_{new}(v_j \to \emptyset, \gamma) &= -\frac{w(v_j, U_i)}{m + \alpha} \\
      &+ \frac{\gamma \cdot (d(v_j) + \alpha) \cdot d(U_i)}{2( m + \alpha)^2}.
    \end{aligned}
    \label{equ:mg_intra_add_vi_former}
  \end{equation}

  We multiply right side of \eqref{equ:mg_intra_add_vi_former} by $2(m+\alpha)^2$ and obtain $Y_{(\text{\ref{equ:mg_intra_add_vi_former}})}$:
  \begin{equation}
    \begin{aligned}
      Y_{(\text{\ref{equ:mg_intra_add_vi_former}})} &= -2(m+\alpha) \cdot w(v_j, U_j)\\
      &+ \gamma \cdot (d(v_j) + \alpha) \cdot d(U_j) \\
      &= X_{(\text{\ref{equ:old_mg_intra_sub_vj}})} + \alpha \cdot (\gamma \cdot d(U_j) - w(v_j, U_j))\\
      &< X_{(\text{\ref{equ:old_mg_intra_sub_vj}})} + \alpha \cdot \gamma \cdot d(U_j)
    \end{aligned}
    \label{equ:mg_intra_add_vi_former_numerator}
  \end{equation}

    Only if $\alpha > \frac{2w(v_j, U_j)}{\gamma \cdot d(U_j)} \cdot m - d(v_j)$, $v_j$ could be removed from its sub-community.

  \underbar{\textbf{Case 3: Consider other vertex $v_k \in S_i$, $k \neq i, j$:}}
  \begin{equation}
    \begin{aligned}
      \Delta M_{new}(v_k \to \emptyset, \gamma) &\leq - \frac{w(v_k, U_k)}{m + \alpha}\\ 
      &+ \frac{\gamma \cdot d(v_k) \cdot (d(U_k) + 2\alpha)}{2(m + \alpha)^2}. \\
    \end{aligned}
    \label{equ:new_mg_intra_add_vk}
  \end{equation}
  The equivalent of \eqref{equ:new_mg_intra_add_vk} holds if and only if $v_k$ is merged after $v_i$ and $v_j$. 
  We multiply right side of \eqref{equ:new_mg_intra_add_vk} by $2(m+\alpha)^2$ and obtain $Y_{(\text{\ref{equ:new_mg_intra_add_vk}})}$:
  \begin{equation}
    \begin{aligned}
      Y_{(\text{\ref{equ:new_mg_intra_add_vk}})} &= -2(m+\alpha) \cdot w(v_k, U_k)\\
      &+ \gamma \cdot d(v_k) \cdot (d(U_k)+ 2\alpha) \\
      &= X_{(\text{\ref{equ:old_mg_intra_sub_vk}})} +  \alpha\cdot\big(2\gamma \cdot d(v_k) - 2 w(v_k, U_k)\big) \\
      &< X_{(\text{\ref{equ:old_mg_intra_sub_vk}})} +  2\alpha\cdot \gamma \cdot d(v_k)
    \end{aligned}
    \label{equ:mg_intra_add_vk_latter_numerator}
  \end{equation}
  
   Only if $\alpha > \frac{w(v_k, U_k)}{\gamma \cdot d(v_k)} \cdot m - \frac{1}{2}d(U_k)$, $v_k$ could be removed from its sub-community. 

  \underbar{\textbf{Case 4: Consider other vertex $v_l \notin S_i$:}}

  \begin{equation}
    \begin{aligned}
      \Delta M_{new}(v_l \to \emptyset, \gamma) &\leq - \frac{w(v_l, U_l)}{m + \alpha}\\ 
      &+ \frac{\gamma \cdot d(v_l) \cdot d(U_l)}{2(m + \alpha)^2}. \\
    \end{aligned}
    \label{equ:new_mg_intra_add_vl}
  \end{equation}
  
  \eqref{equ:new_mg_intra_add_vl} holds if and only if $v_l$ is merged after $v_i$ and $v_j$.
  We multiply right side of \eqref{equ:new_mg_intra_add_vl} by $2(m+\alpha)^2$ and obtain $Y_{(\text{\ref{equ:new_mg_intra_add_vl}})}$:
  
  \begin{equation}
    \begin{aligned}
      Y_{(\text{\ref{equ:new_mg_intra_add_vl}})} &= -2(m+\alpha) \cdot w(v_l, U_l)\\
      &+ \gamma \cdot d(v_l) \cdot d(U_l) \\
      &= X_{(\text{\ref{equ:old_mg_intra_sub_vl}})} -  2 \alpha \cdot w(v_l, U_l) < 0
    \end{aligned}
    \label{equ:mg_intra_add_vk_latter_numerator}
  \end{equation}
  $v_l$ is not affected by the intra-sub-community insertion.
}

{
  Now, we consider the \textbf{insertion of cross-sub-community edges}.
  We adopt the same notations as in the proof of Lemma \ref{lem:cross_sub_e_del}.
  Based on this setup, the modularity gain after the edge insertion is shown as follows.

  \underbar{\textbf{Case 5: Consider the endpoint $v_i$:}}
  \begin{equation}
    \begin{aligned}
      \Delta M_{new}(v_{i} \to \emptyset, \gamma) &= -\frac{w(v_i, U_i)}{m + \alpha} \\
      &+ \frac{\gamma \cdot (d(v_i)+ \alpha) \cdot d(U_i)}{2( m + \alpha)^2}.
    \end{aligned}
    \label{equ:mg_cross_add_vj_latter}
  \end{equation}

  We multiply right side of \eqref{equ:mg_cross_add_vj_latter} by $2(m+\alpha)^2$ and obtain $Y_{(\text{\ref{equ:mg_cross_add_vj_latter}})}$:
  \begin{equation}
    \begin{aligned}
      Y_{(\text{\ref{equ:mg_cross_add_vj_latter}})} &= -2(m+\alpha) \cdot w(v_i, U_i)\\
      &+ \gamma \cdot (d(v_i)+ \alpha) \cdot d(U_i) \\
      &= X_{(\text{\ref{equ:old_mg_intra_sub_vi}})} +  \alpha\cdot\big(\gamma \cdot d(U_i) - 2 w(v_i, U_i)\big) \\
      &< X_{(\text{\ref{equ:old_mg_intra_sub_vi}})} +  \alpha\cdot \gamma \cdot d(U_i)
    \end{aligned}
    \label{equ:mg_cross_add_vj_latter_numerator}
  \end{equation}
  Only if $\alpha > \frac{2w(v_i, U_i)}{\gamma \cdot d(U_i)} \cdot m - d(v_i)$, $v_i$ could be removed from its sub-community. 
  $v_j$ is the same.

  \underbar{\textbf{Case 6: Consider other vertex $v_k \in S_i\cup S_j$, $k \neq i, j$:}}
  \begin{equation}
    \begin{aligned}
      \Delta M_{new}(v_k \to \emptyset, \gamma) &\leq - \frac{w(v_k, U_k)}{m + \alpha}\\ 
      &+ \frac{\gamma \cdot d(v_k) \cdot (d(U_k) + \alpha)}{2(m + \alpha)^2}. \\
    \end{aligned}
    \label{equ:new_mg_cross_add_vk}
  \end{equation}

  The equivalent of \eqref{equ:new_mg_cross_add_vk} holds if and only if $v_k$ is merged after $v_i$ or $v_j$. 
  We multiply right side of \eqref{equ:new_mg_cross_add_vk} by $2(m+\alpha)^2$ and obtain $Y_{(\text{\ref{equ:new_mg_cross_add_vk}})}$:
  \begin{equation}
    \begin{aligned}
      Y_{(\text{\ref{equ:new_mg_cross_add_vk}})} &= -2(m+\alpha) \cdot w(v_k, U_k)\\
      &+ \gamma \cdot d(v_k) \cdot (d(U_k) + \alpha) \\
      &= X_{(\text{\ref{equ:old_mg_intra_sub_vk}})} +  \alpha\cdot\big(\gamma \cdot d(v_k) - 2 w(v_k, U_k)\big) \\
      &< X_{(\text{\ref{equ:old_mg_intra_sub_vk}})} +  \alpha\cdot \gamma \cdot d(v_k) \\
      &< X_{(\text{\ref{equ:old_mg_intra_sub_vk}})} +  2\alpha\cdot \gamma \cdot d(v_k)
    \end{aligned}
    \label{equ:new_mg_cross_add_vk_numerator}
  \end{equation}
  
  $v_k$ could be removed from its sub-community only if $\alpha > \frac{w(v_k, U_k)}{\gamma d(v_k)} \cdot m - \frac{1}{2}d(U_k)$. 

  \underbar{\textbf{Case 7: Consider other vertex $v_l \notin S_i \cup S_j$:}}

  \begin{equation}
    \begin{aligned}
      \Delta M_{new}(v_l \to \emptyset, \gamma) &\leq - \frac{w(v_l, U_l)}{m + \alpha}\\ 
      &+ \frac{\gamma \cdot d(v_l) \cdot d(U_l)}{2(m + \alpha)^2}. \\
    \end{aligned}
    \label{equ:new_mg_cross_add_vl}
  \end{equation}
  
  \eqref{equ:new_mg_cross_add_vl} holds if and only if $v_l$ is merged after $v_i$ and $v_j$.
  We multiply right side of \eqref{equ:new_mg_cross_add_vl} by $2(m+\alpha)^2$ and obtain $Y_{(\text{\ref{equ:new_mg_cross_add_vl}})}$:
  
  \begin{equation}
    \begin{aligned}
      Y_{(\text{\ref{equ:new_mg_cross_add_vl}})} &= -2(m+\alpha) \cdot w(v_l, U_l)\\
      &+ \gamma \cdot d(v_l) \cdot d(U_l) \\
      &= X_{(\text{\ref{equ:old_mg_intra_sub_vl}})} -  2 \alpha \cdot w(v_l, U_l) < 0
    \end{aligned}
    \label{equ:mg_cross_add_vl_latter_numerator}
  \end{equation}
    $v_l$ is not affected by the cross-sub-community insertion.

  Conclusively, the effects of these edge insertions are:
    \begin{enumerate}
        \item $v_i$ could be removed from its sub-community only if $\alpha > \frac{4}{\gamma}m - d(I_i)$ or $\alpha > \frac{2w(v_i, U_i)}{\gamma \cdot d(U_i)} \cdot m - d(v_i)$ according to \textbf{Case 1 and 5}. 
        
        \item  $v_j$ could be removed from its sub-community, only if $\alpha > \frac{2w(v_j, U_j)}{\gamma \cdot d(U_j)} \cdot m - d(v_j)$ according to \textbf{Case 2 and 5}. 
        
        \item $v_k \in S_i \cup S_j$ ($k \neq i, j$) could be removed from its sub-community only if $\alpha > \frac{w(v_k, U_k)}{\gamma \cdot d(v_k)} \cdot m - \frac{1}{2}d(U_k)$ according to \textbf{Case 3 and 6}. 

        \item $v_l \notin S_i \cup S_j$ is unaffected according to \textbf{Case 4 and 7}.
    \end{enumerate}
}

\end{proof}

\section{Additional experiments}
\label{sec:append:exp}
\input{sections/Appendices/exp/gamma_fig}
We now present our additional experimental results. 
First, we evaluate the effect of $\gamma$ on modularity and runtime. 
Next, we show how the number of vertices and edges changes with the number of iterations. 
We then present the numerical results for \texttt{AFF}. 
Finally, we demonstrate the long-term validity of \texttt{HIT-Leiden}.

$\bullet$ {\bf Effect of $\gamma$ on modularity and runtime.}
Fig. \ref{fig:exp_mod_gm} shows the average modularity values for all maintenance algorithms, with the parameter $\gamma\in \{0.5, 2, 8, 32\}$ across all 9 batches, and with the batch size fixed at 1000.
Across all datasets, our {\tt HIT-Leiden} still achieves comparable modularity with other methods across different $\gamma$. 
Fig. \ref{fig:exp_time_gm} shows the average runtime across all maintenance algorithms under the identical configuration.
Across all datasets, the runtime of \texttt{HIT-Leiden} increases as $\gamma$ increases.
This occurs because higher $\gamma$ values tend to form smaller, finer-grained communities, leading to an increase in the number of super-vertices that change their community or sub-community affiliation.

$\bullet$ {\bf Overall results for the number of vertices and edges.} 
Fig. \ref{fig:vesize} displays the overall results for the average number of vertices and edges in \texttt{HIT-Leiden} with respect to the number of iterations, with $b$ and $P$ set to their default values $1,000$ and $10$, respectively.
Clearly, the first-level graph dominates the total size, and thus \texttt{HIT-Leiden} consumes asymptotically similar space to \texttt{Leiden}, whose space complexity is $O\left(|V| + |E|\right)$.

$\bullet$ {\bf Overall results for \texttt{AFF}.} 
Fig. \ref{fig:exp_overal_aff} displays the overall results for the average number of supervertices that change their communities or sub-communities (i.e., $|\texttt{AFF}|$), with $b$ set to its default value $1,000$.
Clearly, {\tt HIT-Leiden} processes the smallest affected region across datasets, particularly on {\tt it-2004} dataset.
This is consistent with the overall efficiency results in Fig. \ref{fig:exp_overal_time}.

$\bullet$ {\bf Long-term effectiveness.}
To demonstrate the long-term effectiveness of maintaining communities, we enlarge the number $r$ of batches from 9 to 999 and set $b=10,000$. 
Fig. \ref{fig:long_term} presents the modularity, proportion of subpartition $\gamma$-dense communities, and runtime on the sx-stackoverflow dataset.
We observe that incremental Leiden algorithms exhibit higher stability than {\tt ST-Leiden} in modularity since they use previous community memberships, and {\tt HIT-Leiden} is faster than other algorithms.
Besides, \texttt{HIT-Leiden} presents its adaptability to substantial edge changes.
Note that when updating the 276\textsuperscript{th} batch, \texttt{HIT-Leiden} incurs high runtime due to substantial changes in community memberships caused by edge evolution, which is also reflected in the corresponding increase in modularity.

\section{Two application studies}
\label{subs:exp:app}
\pgfplotstableread[row sep=\\,col sep=&]{
  size	&	st	&	hit	\\
1	&	288341	&	552459	\\
2	&	290897	&	55583.1	\\
3	&	290727	&	37452.4	\\
4	&	291312	&	37398.2	\\
5	&	291389	&	37321.2	\\
6	&	291088	&	37577.2	\\
7	&	290846	&	37791.2	\\
8	&	292406	&	38150	\\
9	&	292215	&	37999.4	\\
10	&	293153	&	39389.9	\\
11	&	293006	&	38171.3	\\
12	&	293488	&	38202.8	\\
13	&	293073	&	38964.4	\\
14	&	293250	&	37988.6	\\
15	&	293670	&	39521.3	\\
16	&	294011	&	39460.2	\\
17	&	295803	&	38130.3	\\
18	&	294694	&	37579.4	\\
19	&	294120	&	37809.9	\\
20	&	295071	&	38029.6	\\
21	&	295312	&	37925	\\
22	&	294672	&	38659.1	\\
23	&	291204	&	38499.8	\\
24	&	290902	&	37289.1	\\
25	&	290931	&	37516.8	\\
26	&	291163	&	36725.5	\\
27	&	291166	&	36821.7	\\
}\appt
\pgfplotstableread[row sep=\\,col sep=&]{
    size	&	st	&	hit	\\
    1	&	0.379222973	&	0.379222973	\\
    2	&	0.37458194	&	0.379598662	\\
    3	&	0.385690789	&	0.383223684	\\
    4	&	0.376326531	&	0.393469388	\\
    5	&	0.392741935	&	0.397580645	\\
    6	&	0.401442308	&	0.393429487	\\
    7	&	0.391859537	&	0.391859537	\\
    8	&	0.390166534	&	0.391752577	\\
    9	&	0.375295043	&	0.393391031	\\
    10	&	0.3875	    &	0.3921875	\\
    11	&	0.383242824	&	0.394103957	\\
    12	&	0.390600924	&	0.397534669	\\
    13	&	0.381679389	&	0.402290076	\\
    14	&	0.384789157	&	0.396837349	\\
    15	&	0.401355422	&	0.396837349	\\
    16	&	0.396706587	&	0.398203593	\\
    17	&	0.398523985	&	0.398523985	\\
    18	&	0.394005848	&	0.398391813	\\
    19	&	0.396376812	&	0.396376812	\\
    20	&	0.385334292	&	0.381020848	\\
    21	&	0.390035587	&	0.39430605	\\
    22	&	0.39321075	&	0.389674682	\\
    23	&	0.384561404	&	0.390175439	\\
    24	&	0.396491228	&	0.38877193	\\
    25	&	0.397089397	&	0.374913375	\\
    26	&	0.393521709	&	0.382494831	\\
    27	&	0.385724091	&	0.376115305	\\
}\appst

\pgfplotstableread[row sep=\\,col sep=&]{
    size	&	strag	&	hitrag	\\
    1	&	16589000   	&	16589000	\\
    2	&	15714000	&	269924.9167\\
    3	&	16817000	&	261794.6481\\
    4	&	15277000	&	321958.6356\\
    5	&	16817000	&	280494.2658\\
    6	&	18806000	&	305698.0984\\
    7	&	15714000	&	269924.9167\\
    8	&	16817000	&	369927.2202\\
    9	&	15277000	&	273177.0241\\
    10	&	18806000	&	305698.0984\\
}\ragtime

\pgfplotstableread[row sep=\\,col sep=&]{
    size	&	strag	&	hitrag	\\
    1	&	26119958	&	26233623	\\
    2	&	26248949	&	80344	\\
    3	&	25153765	&	266972	\\
    4	&	26198243	&	411849	\\
    5	&	26380166	&	208077	\\
    6	&	26119958	&	304487	\\
    7	&	26366384	&	97295	\\
    8	&	26056902	&	129716	\\
    9	&	25433178	&	295121	\\
    10	&	26038671	&	181083	\\
}\ragtc

\pgfplotstableread[row sep=\\,col sep=&]{
    size	&	strag	&	hitrag	\\
    1	&	0.566   	&	0.566	\\
    2	&	0.57    	&	0.566	\\
    3	&	0.551   	&	0.557	\\
    4	&	0.566   	&	0.569	\\
    5	&	0.55     	&	0.562	\\
    6	&	0.566   	&	0.568   \\
    7	&	0.551    	&	0.56	\\
    8	&	0.562   	&	0.567   \\
    9	&	0.574   	&	0.565   \\
    10	&	0.55       	&	0.575	\\
}\ragacc

\begin{figure*}[t]
    \centering
    \small

    \begin{minipage}[b]{0.57\linewidth}
        \centering
        \begingroup
        \setlength{\textwidth}{\linewidth}
        \pgfplotslegendfromname{rag}\\[-2pt]
        \subfloat[Runtime]{
        \begin{minipage}[b]{0.327290969899666\textwidth}
            \centering
            \begin{tikzpicture}[scale=0.38]
                \begin{axis}[
                    xtick={1,2,3,4,5,6,7,8,9,10},
                    xticklabels={0,1,2,3,4,5,6,7,8,9},
                    xmin=0.6,xmax=10.5,
                    ymin=8*10^4,ymax=2*10^8,
                    ymode = log,
                    width=2.5\textwidth,
                    height=1.8\textwidth,
                    mark size=6pt,
                    line width=2.5pt,
                    ylabel={\LARGE \bf Runtime (ms)},
                    ylabel style={yshift=5pt},
                    xlabel={\LARGE \bf batch},
                    ticklabel style={font=\LARGE},
                    every axis plot/.append style={ultra thick},
                    every axis/.append style={ultra thick},
                ]
                    \addplot [mark=o, line width=2.0pt, color=c2] table[x=size,y=strag]{\ragtime};
                    \addplot [mark=triangle,line width=2.0pt,color=c7,mark options={rotate=180}] table[x=size,y=hitrag]{\ragtime};
                \end{axis}
            \end{tikzpicture}
        \end{minipage}
        \label{fig:rag:time}
        }
        \subfloat[Token cost]{
        \begin{minipage}[b]{0.327290969899666\textwidth}
            \centering
            \begin{tikzpicture}[scale=0.38]
                \begin{axis}[
                    xtick={1,2,3,4,5,6,7,8,9,10},
                    xticklabels={0,1,2,3,4,5,6,7,8,9},
                    xmin=0.6,xmax=10.5,
                    ymin=5*10^3,ymax=2*10^9,
                    ymode = log,
                    mark size=6pt,
                    line width=2.5pt,
                    width=2.5\textwidth,
                    height=1.8\textwidth,
                    ylabel={\LARGE \bf \# of tokens},
                    ylabel style={yshift=5pt},
                    xlabel={\LARGE \bf batch},
                    ticklabel style={font=\LARGE},
                    every axis plot/.append style={ultra thick},
                    every axis/.append style={ultra thick},
                ]
                    \addplot [mark=o, line width=2.0pt, color=c2] table[x=size,y=strag]{\ragtc};
                    \addplot [mark=triangle,line width=2.0pt,color=c7,mark options={rotate=180}] table[x=size,y=hitrag]{\ragtc};
                \end{axis}
            \end{tikzpicture}
        \end{minipage}
        \label{fig:rag:token}
        }
        \subfloat[Accuracy]{
        \begin{minipage}[b]{0.327290969899666\textwidth}
            \centering
            \begin{tikzpicture}[scale=0.38]
                \begin{axis}[
                    legend style = {
                        legend columns=-1,
                        inner sep=0pt,
                        draw=none
                    },
                    legend image post style={scale=1.2, line width=0.8pt},
                    legend to name=rag,
                    xtick={1,2,3,4,5,6,7,8,9,10},
                    xticklabels={0,1,2,3,4,5,6,7,8,9},
                    xmin=0.6,xmax=10.5,
                    ymin=0.5,ymax=0.6,
                    mark size=6pt,
                    line width=2.5pt,
                    width=2.5\textwidth,
                    height=1.8\textwidth,
                    ylabel={\LARGE \bf Accuracy},
                    ylabel style={yshift=5pt},
                    xlabel={\LARGE \bf batch},
                    ticklabel style={font=\LARGE},
                    every axis plot/.append style={ultra thick},
                    every axis/.append style={ultra thick},
                ]
                    \addlegendimage{mark=o, line width=2.0pt, color=c2}
                    \addlegendentry{{\small \tt ST-Leiden-RAG}}

                    \addlegendimage{mark=triangle,line width=2.0pt,color=c7,mark options={rotate=180}}
                    \addlegendentry{{\small \tt HIT-Leiden-RAG}}

                    \addplot [mark=o, line width=2.0pt, color=c2] table[x=size,y=strag]{\ragacc};
                    \addplot [mark=triangle,line width=2.0pt,color=c7,mark options={rotate=180}] table[x=size,y=hitrag]{\ragacc};
                \end{axis}
            \end{tikzpicture}
        \end{minipage}
        \label{fig:rag:acc}
        }
        \caption{Comparison between \texttt{ST-Leiden-RAG} and \texttt{HIT-Leiden-RAG} over 9 update batches on Graph-RAG.}
        \label{fig:rag}

        \endgroup
    \end{minipage}%
    \hfill
    \begin{minipage}[b]{0.38\linewidth}
        \centering
        \begingroup
        \setlength{\textwidth}{\linewidth}
        \small
        \pgfplotslegendfromname{risk}\\[-1pt]
        \subfloat[Runtime]{
        \begin{minipage}[b]{0.494879622344611\textwidth}
            \centering
            \begin{tikzpicture}[scale=0.38]
                \begin{axis}[
                    xtick={1, 10, 19, 27},
                    xticklabels={$0$, $9$,$18$,$26$},
                    xmin=0,xmax=30,
                    ymin=5*10^3,ymax=2*10^6,
                    ymode = log,
                    mark size=4pt,
                    line width=2.5pt,
                    width=2.5\textwidth,
                    height=1.8\textwidth,
                    xlabel={\LARGE \bf batch},
                    ylabel={\LARGE \bf Runtime (ms)},
                    ylabel style={yshift=5pt},
                    ticklabel style={font=\LARGE},
                    every axis plot/.append style={ultra thick},
                    every axis/.append style={ultra thick},
                ]
                    \addplot [smooth, line width=2.0pt, color=c2] table[x=size,y=st]{\appt};
                    \addplot [smooth, line width=2.0pt, color=c7] table[x=size,y=hit]{\appt};
                \end{axis}
            \end{tikzpicture}
        \end{minipage}
        \label{fig:risk:runtime}
        }
        \subfloat[Recall]{
        \begin{minipage}[b]{0.494879622344611\textwidth}
            \centering
            \begin{tikzpicture}[scale=0.38]
                \begin{axis}[
                    legend style = {
                        legend columns=-1,
                        inner sep=0pt,
                        draw=none
                    },
                    legend image post style={scale=1.2, line width=0.8pt},
                    legend to name=risk,
                    xtick={1, 10, 19, 27},
                    xticklabels={$0$, $9$,$18$,$26$},
                    xmin=0,xmax=30,
                    ytick={0.36, 0.39, 0.42},
                    yticklabels = {$0.36$,$0.39$, $0.42$},
                    ymin=0.35,ymax=0.43,
                    width=2.5\textwidth,
                    height=1.8\textwidth,
                    mark size=4pt,
                    line width=2.5pt,
                    xlabel={\LARGE \bf batch},
                    ylabel={\LARGE \bf Recall},
                    ylabel style={yshift=5pt},
                    ticklabel style={font=\LARGE},
                    every axis plot/.append style={ultra thick},
                    every axis/.append style={ultra thick},
                ]
                    \addplot [smooth, line width=2.0pt, color=c2] table[x=size,y=st]{\appst};
                    \addplot [smooth, line width=2.0pt, color=c7] table[x=size,y=hit]{\appst};

                    \addlegendimage{smooth, line width=2.0pt, color=c2}
                    \addlegendentry{{\small \tt ST-Leiden}}

                    \addlegendimage{smooth,line width=2.0pt,color=c7,mark options={rotate=180}}
                    \addlegendentry{{\small \tt HIT-Leiden}}
                \end{axis}
            \end{tikzpicture}
        \end{minipage}
        \label{fig:risk:recall}
        }
        \caption{Comparison between \texttt{ST-Leiden} and \texttt{HIT-Leiden} over 26 update batches on fraud-ring discovery.}
        \label{fig:risk}
        \endgroup
    \end{minipage}
\end{figure*}
In this section, we apply \texttt{HIT-Leiden} to two real applications and analyze its empirical results.

$\bullet$ {\bf Graph-RAG \cite{graphrag2025}.}
%
To improve the LLM generation for answering a question, people often first retrieve relevant information from an external corpus and then augment the question with this information before invoking the LLM.
To facilitate retrieval, Graph-RAG builds an offline index: It first builds a graph for the corpus, then clusters the graph hierarchically using Leiden, and finally associates a summary for each community, which is generated by an LLM with some token cost.
In practice, since the underlying corpus often changes, the communities and their summaries need to be updated as well.
Our {\tt HIT-Leiden} can not only dynamically update the communities efficiently, but also save the token cost since we only need to regenerate the summaries for the updated communities.

To experiment, we use the HotpotQA \cite{yang2018hotpotqa} dataset, which contains Wikipedia-based question-answer (QA) pairs.
We randomly select 9,500 articles to build the initial graph, and insert 9 batches of new articles, each with 5 articles.
The LLM we use is Doubao-1.5-pro-32k.
To support a dynamic corpus, we adapt the static Graph-RAG method by updating communities using {\tt ST-Leiden} and {\tt HIT-Leiden}, respectively.
These two RAG methods are denoted by {\tt ST-Leiden-RAG} and {\tt HIT-Leiden-RAG}, respectively.
Note that {\tt ND-Leiden}, {\tt DS-Leiden}, and {\tt DF-Leiden} are not fit to maintain the hierarchical communities of Graph-RAG since they lack hierarchical maintenance.
We report their runtime, token cost, and accuracy in Fig. \ref{fig:rag}.
Clearly, {\tt HIT-Leiden-RAG} is $56.1 \times$ faster than {\tt ST-Leiden-RAG}.
Moreover, it significantly reduces the summary token cost while preserving downstream QA accuracy, since its token cost is only 0.8\% of the token cost of {\tt ST-Leiden-RAG}.
Hence, {\tt HIT-Leiden} is effective for supporting Graph-RAG on a dynamic corpus.

$\bullet$ \textbf{Fraud-ring discovery \cite{liu2022eriskcom}.}
We also experiment with a workflow of fraud-ring discovery at \Acom, which is different from the risk dataset used in the general maintenance benchmark in Table \ref{tab:datasets}.
A fraud ring is a group with strong internal coordination that commits fraud. 
Many fraud cases are group-based, so fraud-ring discovery can be viewed as finding densely connected communities in a graph.
Leiden partitions the graph into candidate communities, reducing search space and scaling efficiently. 
{
Based on Leiden communities, downstream modules produce a list of high-risk individuals. 
The list is then manually reviewed, and the recall of confirmed ring members is reported.
}
As the graph evolves, we need to update Leiden communities to timely discover new potential ring members.

The input graph contains 58.3M vertices and 102.9M edges spanning 206 days, with daily updates (about 0.5M edge updates per batch).
We first construct a static graph using the updates within the first 180 days.
Then, we maintain the graph with a sliding window of 180 days; that is, we update it day by day by adding edges from the remaining 26 days while removing edges that fall outside the window.
%
%
Finally, we evaluate the runtime of {\tt ST-Leiden} and {\tt HIT-Leiden}, since other processes are black-box to us, and the recall, as shown in Fig. \ref{fig:risk}.
Clearly, {\tt HIT-Leiden} achieves $7.7\times$ the speed of {\tt ST-Leiden} with comparable recall.

\end{document}